\documentclass[journal,twocolumn]{IEEEtran}
\usepackage[T1]{fontenc}
\usepackage{cite}
\usepackage{amsmath,amssymb}
\usepackage{bm}
\usepackage{graphicx}
\usepackage{subfigure}
\usepackage{textcomp}
\usepackage{xcolor}
\usepackage{threeparttable}
\usepackage{multirow}
\usepackage{verbatim}
\usepackage{url}
\usepackage{hyperref}
\usepackage{comment}
\usepackage{color}
\usepackage{booktabs}
\usepackage{amsthm}
\usepackage{makecell}
\usepackage{cuted}

\usepackage[linesnumbered, ruled]{algorithm2e}
\SetKwRepeat{Do}{do}{while}%

\newtheorem{lemma}{Lemma}

\def\BibTeX{{\rm B\kern-.05em{\sc i\kern-.025em b}\kern-.08em
    T\kern-.1667em\lower.7ex\hbox{E}\kern-.125emX}}
\begin{document}
\title{Low-Complexity Joint Beamforming for RIS-Assisted MU-MISO Systems Based on Model-Driven Deep Learning}

\author{
{Weijie Jin, \IEEEmembership{Graduate Student Member, IEEE}, Jing Zhang, \IEEEmembership{Member, IEEE}, \\Chao-Kai Wen, \IEEEmembership{Fellow, IEEE}, Shi Jin, \IEEEmembership{Fellow, IEEE},\\ Xiao Li, \IEEEmembership{Member, IEEE}, Shuangfeng Han, \IEEEmembership{Senior Member, IEEE}}
\vspace{-10pt}

\thanks{This article was presented in part at the IEEE ICCT 2023 \cite{jinJointBeamforming2023}.}

\thanks{W. Jin, J. Zhang, S. Jin, and X. Li are with the National Mobile Communications Research Laboratory, Southeast University, Nanjing 210096, China (email: \{jinweijie, jingzhang, jinshi, li\_xiao\}@seu.edu.cn).

C.-K. Wen is with the Institute of Communications Engineering, National Sun Yat-sen University, Kaohsiung 80424, Taiwan (email: chaokai.wen@mail.nsysu.edu.tw).

S. Han is with the China Mobile Research Institute, Beijing 100053, China
(email: hanshuangfeng@chinamobile.com).}
}
\maketitle

\begin{abstract}
    Reconfigurable intelligent surfaces (RIS) can improve signal propagation environments by adjusting the phase of the incident signal. However, optimizing the phase shifts jointly with the beamforming vector at the access point is challenging due to the non-convex objective function and constraints. In this study, we propose an algorithm based on weighted minimum mean square error optimization and power iteration to maximize the weighted sum rate (WSR) of a RIS-assisted downlink multi-user multiple-input single-output system. To further improve performance, a model-driven deep learning (DL) approach is designed, where trainable variables and graph neural networks are introduced to accelerate the convergence of the proposed algorithm. We also extend the proposed method to include beamforming with imperfect channel state information and derive a two-timescale stochastic optimization algorithm. Simulation results show that the proposed algorithm outperforms state-of-the-art algorithms in terms of complexity and WSR. Specifically, the model-driven DL approach has a runtime that is approximately 3\% of the state-of-the-art algorithm to achieve the same performance. Additionally, the proposed algorithm with 2-bit phase shifters outperforms the compared algorithm with continuous phase shift.
\end{abstract}


\section{Introduction}
\IEEEPARstart{T}{he} pursuit of enhanced spectral efficiency continues to shape the evolution of communication systems. Over the past few decades, wireless communication has experienced remarkable advancements propelled by various technologies, including multiple-input multiple-output (MIMO), orthogonal frequency division multiplexing (OFDM), and millimeter-wave solutions. However, these technologies have primarily focused on elevating data rates through spectrum resource optimization and multiplexing, often without treating the channel as a malleable element in the design process. Traditional physical-layer communication systems typically optimize system performance based on specific channel conditions to achieve higher spectral efficiency. In recent years, research on metasurfaces has made substantial strides, leading to the emergence of reconfigurable intelligent surfaces (RIS) that can dynamically manipulate communication environments. This introduces the channel as a new dimension for design, holding the potential to enhance spectral efficiency while minimizing energy consumption and hardware costs \cite{wuIntelligentReflectingSurfaceAided2021, wuIntelligentReflectingSurface2019, wuBeamformingOptimizationWireless2020, huangReconfigurableIntelligentSurfaces2019, chenActiveIRSAided2023, chenChannelEstimationTraining2023a}.

\subsection{Related Work}
\vspace{7pt}
A RIS comprises numerous low-cost reflective elements that manipulate incident signals through predefined phase shifts. Careful phase shift design within passive beamforming yields an asymptotic power gain on the order of $\mathcal{O}(N^2)$ \cite{wuIntelligentReflectingSurface2019}, considerably more efficient compared to relays with the order of $\mathcal{O}(N)$. Nevertheless, RIS deployment introduces its own set of challenges. Designing joint active beamforming at the access point (AP) and passive beamforming at the RIS is complex due to the non-convex nature of the objective function and associated constraints.
\vspace{7pt}

Early studies primarily delved into joint beamforming in single-user scenarios \cite{wuIntelligentReflectingSurface2018, yuMISOWirelessCommunication2019, yuOptimalBeamformingMISO2020, fengJointBeamformingOptimization2021, fengDeepReinforcementLearning2020}. The joint beamforming task is often decomposed into separate active and passive beamforming designs. In single-user scenarios, maximum ratio transmission is commonly utilized to compute active beamforming. In \cite{wuIntelligentReflectingSurface2018}, passive beamforming was derived using semidefinite relaxation (SDR) to maximize received signal power at user terminals. An alternative approach in \cite{yuMISOWirelessCommunication2019} exploited fixed-point iteration for passive beamforming design, offering improved performance and lower complexity compared to SDR.
\vspace{7pt}

Extending to multi-user scenarios aligns better with wireless communication trends \cite{wuIntelligentReflectingSurface2019, wuBeamformingOptimizationWireless2020, guoWeightedSumRateMaximization2019, guoWeightedSumRateMaximization2020, chenFundamentalLimitsIntelligent2023, chenActiveIRSAided2023a}. Here, the minimum mean square error (MMSE) criterion and SDR method were employed to solve for active and passive beamforming, respectively \cite{wuIntelligentReflectingSurface2019}. To optimize the system's weighted sum rate (WSR), the Lagrangian dual transform was harnessed, with passive beamforming updates obtained via closed-form solutions \cite{guoWeightedSumRateMaximization2019}. Employing fractional programming, \cite{guoWeightedSumRateMaximization2020} broke down joint beamforming into four separate blocks and merged them with the block coordinate descent (BCD) method for obtaining a stationary solution. The capacity region of channels was investigated through a bisection-based framework in \cite{chenFundamentalLimitsIntelligent2023}. However, these methods suffer from compromised performance and high complexity, impeding practical deployment.

Recent years have witnessed increasing integration of deep learning (DL) into physical layer communications, encompassing channel estimation \cite{jinAdaptiveChannelEstimation2021, heBeamspaceChannelEstimation2022}, signal detection \cite{zhouModelDrivenDeepLearningBased2021, zhangMetaLearningBasedMIMO2021}, and channel state information (CSI) feedback \cite{wenDeepLearningMassive2018, guoDeepLearningBasedTwoTimescale2022}. In joint beamforming design, DL has been employed to directly learn the mapping from CSI to beamforming vectors \cite{xuRobustDeepLearningBased2022, huangReconfigurableIntelligentSurface2020}. A neural network composed of fully connected layers was utilized to simultaneously produce active and passive beamforming in \cite{xuRobustDeepLearningBased2022}. Deep reinforcement learning treated the system as an agent, training a deep deterministic policy gradient neural network via trial-and-error interactions with the environment \cite{huangReconfigurableIntelligentSurface2020}. However, these endeavors predominantly rely on data-driven DL, susceptible to performance degradation with changing environments and lacking interpretability. Model-driven DL, incorporating expert knowledge into network design, has shown robustness and interpretability in \cite{jinAdaptiveChannelEstimation2021, heBeamspaceChannelEstimation2022, zhouModelDrivenDeepLearningBased2021}. Yet, as far as our knowledge extends, a model-driven DL-based joint beamforming design for RIS-assisted systems remains to be investigated, despite being an emerging trend.

The aforementioned studies assume perfect CSI availability, which is not always achievable. Addressing performance degradation in imperfect CSI scenarios, bounded and statistical CSI error models were proposed \cite{zhouFrameworkRobustTransmission2020}. In \cite{huRobustSecureSumRate2021}, the authors investigated robust and secure communication aided by a self-sustaining RIS under the bounded CSI error model. To handle challenges in estimating RIS-related channels, \cite{guoWeightedSumRateMaximization2020} introduced a two-timescale transmission structure and utilized stochastic successive convex approximation (SSCA) algorithm \cite{liuOnlineSuccessiveConvex2018} to maximize average WSR. By relaxing the unimodular constraint, the SSCA algorithm was enhanced to accelerate convergence \cite{zhaoIntelligentReflectingSurface2021}. However, the SSCA algorithm necessitates designing a surrogate function to solve the non-convex optimization, which is not an equivalent transformation and might cause performance loss. Certain works circumvented the use of instantaneous CSI in algorithm design, alleviating the burden of channel estimation \cite{hanLargeIntelligentSurfaceAssisted2019, ganRISAssistedMultiUserMISO2021, jiangLearningReflectBeamform2021}. Specifically, statistical CSI was applied to maximize average received signal-to-noise ratio (SNR) in single-user systems \cite{hanLargeIntelligentSurfaceAssisted2019} and ergodic capacity in multi-user systems \cite{ganRISAssistedMultiUserMISO2021}. A graph neural network (GNN) in \cite{jiangLearningReflectBeamform2021} directly learned the mapping from received pilots to beamforming vectors, circumventing the need for channel estimation. However, these works' performance could be further improved with reduced complexity, motivating the design of a low-complexity algorithm and its extension to imperfect CSI scenarios.

\subsection{Contributions and Organization}
In this study, we delve into the intricacies of joint beamforming in RIS-assisted multi-user multi-input single-output (MU-MISO) systems. We present an iterative algorithm and subsequently introduce a model-driven DL-based neural network to further streamline complexity. Moreover, we extend the model-driven DL approach to account for imperfect CSI scenarios, showcasing improved performance relative to existing algorithms. Our contributions are as follows:

\begin{itemize}
  \item \emph{Development of an Iterative Algorithm for RIS-Assisted MU-MISO Systems:} The intricate challenge of joint beamforming is addressed by dissecting it into active and passive beamforming components. This decomposition is achieved by transforming the optimization problem into a weighted MMSE (WMMSE) optimization form. For active beamforming, we derive a closed-form solution, while passive beamforming's unimodular quadratic program is tackled using the power iteration (PI) algorithm. The result is an iterative algorithm grounded in the alternating optimization principle, ensuring convergence.

  \item \emph{Low-Complexity Model-Driven DL with Graph Neural Network:} To further mitigate algorithmic complexity, we fuse model-driven DL into the process to capture crucial parameters linked to convergence speed. We devise a GNN that takes into account the optimal active beamforming structure, enhancing initialization performance. The GNN encapsulates user interactions through its permutation invariant/equivariant properties. Additionally, to accommodate discrete phase shifts, we introduce an approximate quantization function. This integration of DL and algorithmic approaches yields a robust and high-performing neural network.

  \item \emph{Enhancements for Imperfect CSI Scenarios:} Acknowledging the substantial overhead associated with channel estimation, we propose a novel transmission protocol that updates passive beamforming only during a frame's initial slot. The objective is to maximize average WSR over the transmission duration, which necessitates a two-timescale optimization framework. Our model-driven DL method is extended to encompass this proposed transmission protocol, bolstered by the introduction of trainable variables to boost performance. 

  \item \emph{Validation through Comprehensive Simulation:} To demonstrate the effectiveness of our proposed algorithms, we rigorously test them under various conditions. Our proposed algorithm attains comparable performance to state-of-the-art approaches while operating at less than 3\% of the runtime. The extended method displays robustness and holds promise for practical deployment in imperfect CSI scenarios, outperforming the existing two-timescale optimization methods.
\end{itemize}

The subsequent sections are structured as follows: Section \ref{Sec_System_Model_and_Problem_Formulation} introduces the system and channel models, followed by the formulation of the two optimizations for both perfect and imperfect CSI scenarios. In Section \ref{Sec_Proposed_Beamforming_Algorithm}, we effectuate a transformation on the optimization problem under perfect CSI and present an iterative algorithm to solve it. In Section \ref{Sec_Model_Driven}, we introduce model-driven DL augmented by a GNN, addressing complexity and performance concerns. We then expand our method to encompass imperfect CSI scenarios in Section \ref{Sec_ImCSI}. The experimental details and numerical results are provided in Section \ref{Sec_Result}. Finally, our work concludes in Section \ref{Sec_Conclusion}.

\subsection{Notations}
The following notations are used throughout this paper: $a$, $\mathbf{a}$, and $\mathbf{A}$ stand for a scalar, a column vector, and a matrix, respectively. ${\mathbf{a}[i:j]}$ represents the $i$-th to the $j$-th elements of $\mathbf{a}$. $(\mathbf{A})_{i,j}$ is the element on the $i$-th row and $j$-th column of $\mathbf{A}$. The conjugate, transpose and conjugate transpose of $\mathbf{A}$ are represented by $\mathbf{A}^{*}$, $\mathbf{A}^T$, and $\mathbf{A}^H$, respectively. $\|\mathbf{A}\|$ represent the Frobenius norm of $\mathbf{A}$. $\mathbf{A}^{-1}$ is the inverse of $\mathbf{A}$. For the symmetric matrix $\mathbf{A}$ and $\mathbf{B}$, $\mathbf{A} \succeq \mathbf{B}$ signifies that $\mathbf{A} - \mathbf{B}$ is positive semidefinite. $\operatorname{diag}(\mathbf{a})$ returns a diagonal matrix with the input vector as its diagonal. $\operatorname{tanh}(\mathbf{a})$ is the hyperbolic tangent function and is applied to each element of the input. $\operatorname{abs}(\cdot)$ and $\operatorname{angle}(\cdot)$ are the function which returns the amplitude and phase of input, respectively. 

\section{System Model and Problem Formulation}
\label{Sec_System_Model_and_Problem_Formulation}
\begin{figure}[tbp!]
    \centering
    \includegraphics[width=0.45\textwidth]{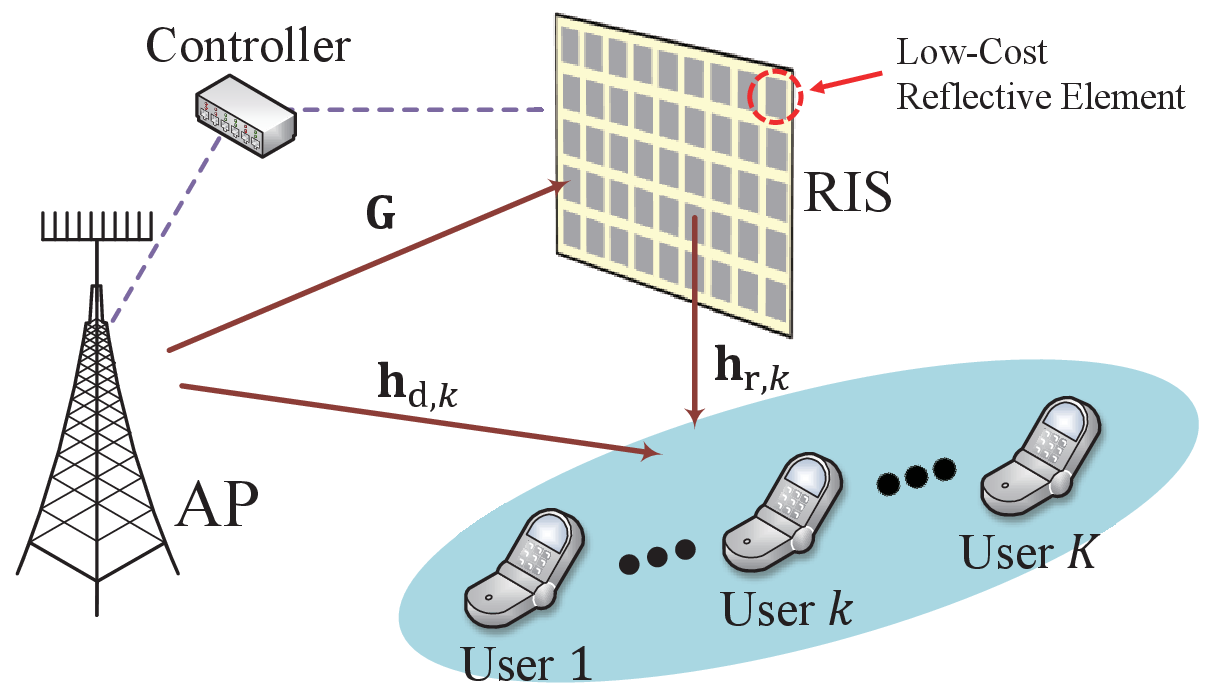}
    \caption{Downlink transmission of the RIS-assisted MU-MISO system.}
    \label{fig_system_model}
\end{figure}

In this section, we begin by providing an overview of the RIS-assisted downlink MU-MISO system. Subsequently, we formulate the joint beamforming problem for both perfect and imperfect CSI scenarios.
\subsection{System Model}
\label{System_Model}
As shown in Fig.~\ref{fig_system_model}, we consider the downlink transmission of the RIS-assisted MU-MISO system. The AP is equipped with $M$ antennas to simultaneously serve $K$ single-antenna users. To improve the WSR of the system, a RIS with $N$ reflective elements is deployed to adjust the communication environments. The system is equipped with a controller which communicates with the AP and RIS for exchanging information such as channel state information and the adopted beamforming vectors.

The channels from AP to RIS, from RIS to user $k$, and from AP to user $k$ are denoted as $\mathbf{G} \in \mathbb{C}^{N\times M}$, $\mathbf{h}_{\mathrm{r},k} \in \mathbb{C}^{1\times N}$, and $\mathbf{h}_{\mathrm{d},k} \in \mathbb{C}^{1 \times M}$, respectively. Assuming that all the channels are quasistatic. The phase shift caused by RIS is denoted as a diagonal matrix $\boldsymbol{\Theta} = \operatorname{diag}(\theta_1, \theta_2, \ldots, \theta_N) \in \mathbb{C}^{N\times N}$, where $\theta_n$ is the phase shift of the $n$-th reflective element of RIS. The amplitude reflection coefficients of the reflective elements are assume to be equal and are omitted, because tunning the amplitude of the incident signal leads to higher energy consumption and hardware cost. The value of $\theta_n$ is limited by the hardware and can be selected from a discrete set. Denote $B$ as the number of phase shifter bits and the number of possible phase shifts as $2^B$. Therefore, the discrete phase shift set can be denoted as $\mathcal{F} = \left\{\theta_n | \theta_n = e^{\jmath \varphi_n}, \varphi_n \in \left\{ 0, \triangle\theta, \ldots, ( 2^B-1)\triangle\theta\right\}\right\}$, where $\triangle\theta = 2\pi/2^B$. Ideally, when $B \rightarrow \infty$, the continuous phase shift is considered.

The signal destined for user $k$ is denoted as $s_k$ and is assumed to be independent, having zero mean and unit variance. The received signal at user $k$ can be expressed as
\begin{equation}
    y_{k} = (\mathbf{h}_{\mathrm{d}, k}+\mathbf{h}_{\mathrm{r} , k} \boldsymbol{\Theta} \mathbf{G}) \sum_{i=1}^{K} \mathbf{w}_{i} s_{i}+n_k,
\end{equation}
where $\mathbf{w}_k\in \mathbb{C}^{M\times 1}$ represents the active beamforming for user $k$, $n_k \sim \mathcal{CN}(0, \sigma_k^2)$ is the received signal noise of user $k$, which follows the circularly symmetric complex Gaussian distribution with mean zero and variance $\sigma_k^2$. By introducing the passive beamforming vector $\boldsymbol{\theta} = \left[\theta_1, \ldots, \theta_N \right]^H \in \mathbb{C}^{1\times N}$ and the cascaded channel $\mathbf{H}_{\mathrm{r} ,k} = \operatorname{diag}(\mathbf{h}_{\mathrm{r},k})\mathbf{G} \in \mathbb{C}^{N\times M}$, the user $k$'s received signal takes the form
\begin{equation}
  y_k = \left(\mathbf{h}_{\mathrm{d},k} + \boldsymbol{\theta}^H\mathbf{H}_{\mathrm{r},k}\right)\sum_{i = 1}^K\mathbf{w}_i s_i + n_k.
\end{equation}
In this expression, signals reflected by the RIS multiple times are disregarded due to extensive path loss, and other user signals are treated as interference to user $k$. The achievable rate of the AP-to-user $k$ channel is given by
\begin{equation}
  \mathcal{R}_{k} = \operatorname{log}\left(1+\frac{\left|\left(\mathbf{h}_{\mathrm{d}, k}+\boldsymbol{\theta}^H \mathbf{H}_{\mathrm{r}, k}\right) \mathbf{w}_{k}\right|^{2}}{\sum_{i=1, i \neq k}^{K}\left|\left(\mathbf{h}_{\mathrm{d}, k}+\boldsymbol{\theta}^H \mathbf{H}_{\mathrm{r}, k}\right) \mathbf{w}_{i}\right|^{2}+\sigma_{k}^{2}}\right).
  \label{equ_channel_capacity}
\end{equation}

\vspace{5pt}
\subsection{Problem Formulation}
\subsubsection{Perfect CSI}
We first address the joint beamforming problem in scenarios with perfect CSI,\footnote{Various conventional techniques for channel estimation can be explored in \cite{mishraChannelEstimationLowcomplexity2019a, jensenOptimalChannelEstimation2020, wangChannelEstimationIntelligent2020}. However, these methods fall outside the scope of this paper and are deferred for future investigation.} aiming to maximize the WSR of the system. The problem can be formulated as
\begin{equation} 
  \begin{aligned}
  ({\bf P1}): \quad \max _{\mathbf{W}, \boldsymbol{\theta}} \,\, & \sum_{k=1}^{K} \alpha_{k} \mathcal{R}_k \\
  \text {s.t.}\,\, & \theta_{n} \in \mathcal{F}, \quad \forall n=1, \ldots, N, \\
  &\sum_{k=1}^{K}\left\|\mathbf{w}_{k}\right\|^{2} \leq P_{\mathrm{T}}.
  \end{aligned}
\end{equation}
Here, $\mathbf{W} = \left[\mathbf{w}_1^T, \mathbf{w}_2^T, \ldots, \mathbf{w}_K^T\right]^T\in \mathbb{C}^{M\times K}$, and $\alpha_k$ represents the priority of user $k$. The elements of the passive beamforming need to be selected from set $\mathcal{F}$, constrained by a unit modulus condition. Active beamforming is bounded by the maximum transmit power constraint: $\sum_{k=1}^{K}\left\|\mathbf{w}_{k}\right\|^{2} \leq P_{\mathrm{T}}$. This optimization task is challenging due to the non-convex nature of the objective function and constraints, with the two optimization variables being deeply interconnected.

\subsubsection{Imperfect CSI}

To optimize the passive beamforming vector, accurate estimation of RIS-related channels $\{\mathbf{G},\mathbf{h}_{\mathrm{r},k}\}$ is essential. However, practically, the number of parameters requiring estimation is usually large due to the large number of reflective elements on the RIS. Additionally, RIS elements often lack the capability for channel estimation. Thus, channel estimation for RIS-related channels is complex and time-consuming, prompting the need for exploring joint beamforming in imperfect CSI scenarios.

\begin{figure}[tbp!]
  \centering
  \includegraphics[width=0.48\textwidth]{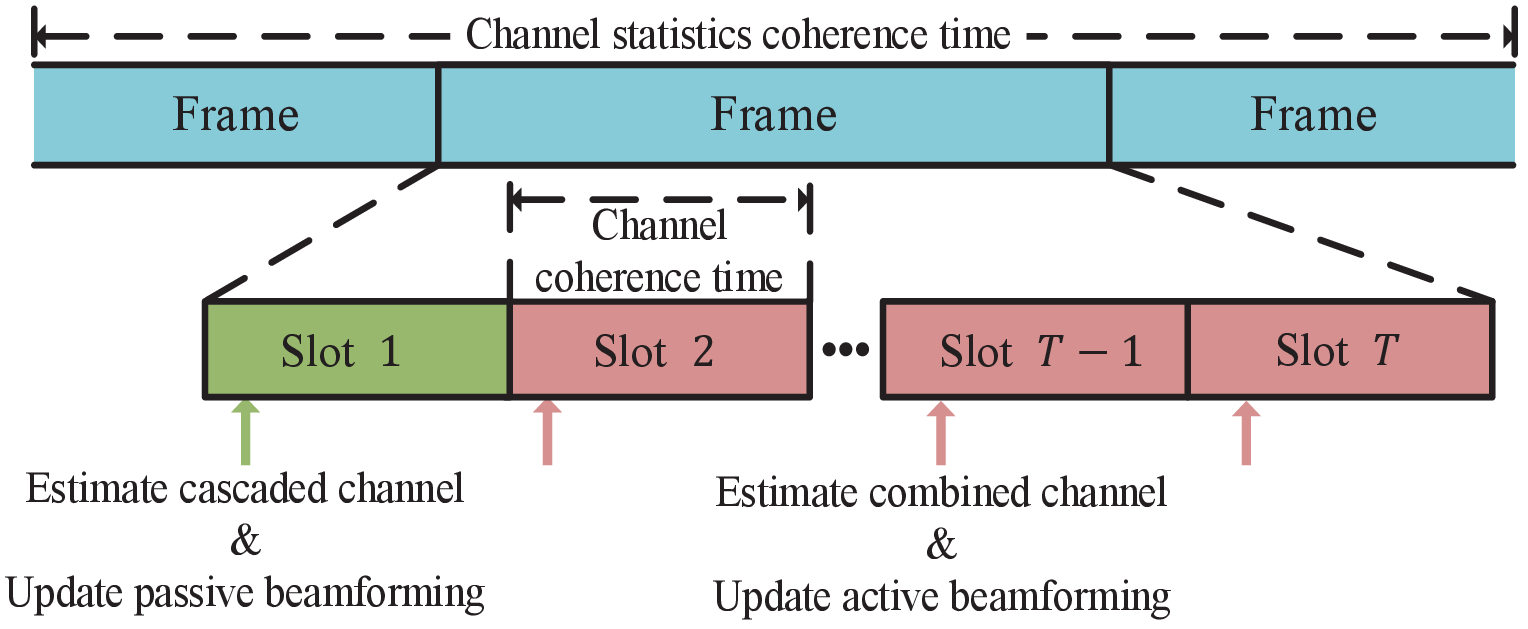}
  \caption{Frame structure of the proposed transmission protocol.}
  \label{fig_frame_structure}
\end{figure}

A frame structure, as depicted in Fig.~\ref{fig_frame_structure}, is considered. Channel statistics are assumed to be invariant within adjacent frame, and the channel is coherent throughout a slot. To reduce substantial channel estimation overhead, cascaded channels are only estimated once in the first slot of each frame. In the remaining slots, only the combined channel $\mathbf{h}_{k} = \mathbf{h}_{\mathrm{d},k} + \boldsymbol{\theta}^H\mathbf{H}_{\mathrm{r},k}$ is estimated. The combined channel is assumed to be estimated accurately. This assumption is reasonable due to the relatively small dimension of the combined channel, allowing conventional methods to be employed. After estimating the combined channel, active beamforming vectors can be updated without requiring cascaded channel information.

However, estimating cascaded channels perfectly with short pilots is difficult. Additionally, these channels tend to change in subsequent slots due to channel variability. Hence, when designing passive beamforming, our aim is to maximize the average WSR in the following frame transmissions. This results in a two-timescale optimization problem:
\begin{equation}
    \begin{aligned}
    ({\bf P2}):\quad \max _{\boldsymbol{\theta}} \,\, & \mathbb{E}_{\xi}{\left[\max _{\mathbf{W}^{[\xi]}} \sum_{k=1}^{K} \alpha_{k} \mathcal{R}_k^{[\xi]} \right]} \\
    \text { s.t.} \,\, & \left|\theta_n\right|\in \mathcal{F}, \quad \forall n=1, \ldots, N,  \\
    &\sum_{k=1}^K\left\|\mathbf{w}_k^{[\xi]}\right\|^2 \leq P_{\mathrm{T}}, \quad \forall \xi.
    \end{aligned}
\end{equation}
In this formulation, $\xi$ signifies the sample index, and $\mathcal{R}_k^{[\xi]}$ is determined by substituting $\{\mathbf{W}^{[\xi]},\mathbf{h}_{\mathrm{d},k}^{[\xi]}, \mathbf{H}_{\mathrm{r},k}^{[\xi]}\}$ into (\ref{equ_channel_capacity}) instead of $\{\mathbf{W}, \mathbf{h}_{\mathrm{d},k}, \mathbf{H}_{\mathrm{r},k}\}$. The terms $\{\mathbf{h}_{\mathrm{d},k}^{[\xi]}, \mathbf{H}_{\mathrm{r},k}^{[\xi]}\}$ represent the estimated channel parameters from previous frames, which are stored in memory and then randomly selected for the passive beamforming design. This problem involves optimizing both long-term variable $\boldsymbol{\theta}$ and short-term optimized variable $\mathbf{W}^{[\xi]}$. Additionally, the expectation operator adds to the problem's complexity.

\section{Proposed Joint Beamforming Algorithm}
\label{Sec_Proposed_Beamforming_Algorithm}

In this section, we initiate the transformation of the optimization problem ${\bf P1}$ into an equivalent WMMSE optimization. Subsequently, we disentangle the active and passive beamforming optimizations. The local optimal solutions for active and passive beamforming are then derived through the WMMSE optimization and the PI algorithm. Ultimately, we introduce the WMMSE-PI algorithm, which hinges on an alternating optimization strategy.

\subsection{Problem Transformation}
The optimization problem ${\bf P1}$ is challenging to be solved because both the objective function and the unit module constraint are non-convex. We first transform the optimization into a WMMSE optimization \cite{shiIterativelyWeightedMMSE2011}, which is written as
\begin{equation}
  \begin{aligned}
  ({\bf P3}): \quad \min _{\mathbf{W}, \boldsymbol{\theta}, \mathbf{u}, \boldsymbol{\lambda}}\,\, &\sum_{k=1}^{K} \alpha_{k} (\lambda_k e_k - \log \lambda_k) \\
  \text {s.t.}\,\, & \theta_{n} \in \mathcal{F}, \quad \forall n=1, \ldots, N.
  \end{aligned}
\end{equation}
where $\mathbf{u} = \{u_1, u_2, \ldots, u_K\}$ and $\boldsymbol{\lambda} = \{\lambda_1, \lambda_2, \ldots, \lambda_K\}$. $u_k$ is the an auxiliary variable, which can be regarded as the receive gain of user $k$. $\lambda_k>0$ is an auxiliary variable without practical meannings. $\alpha_k$ is the priority of user $k$. $e_k$ is the equivalent mean square error of the received signal $s_k$, which can be written as
\begin{equation}
  \begin{aligned}
    e_k = & 1 - u_k^*\mathbf{h}_k \mathbf{w}_k - u_k \mathbf{w}_k^H \mathbf{h}_k^H + u_k u_k^* \mathbf{h}_k \sum_{i = 1}^K\mathbf{w}_i \mathbf{w}_i^H \mathbf{h}_k^H\\
    & + \frac{\sigma_k^2}{P_\mathrm{T}} u_k u_k^* \sum_{i = 1}^K\mathbf{w}_i^H \mathbf{w}_i.
  \end{aligned}
  \label{equ_equivalent_MSE}
\end{equation}
An auxiliary term $ \sum_{i = 1}^K (\mathbf{w}_i^H \mathbf{w}_i)/P_\mathrm{T}$ is introduced to the last term of (\ref{equ_equivalent_MSE}) to remove the maximum transmit power constraint in ${\bf P1}$. The equivalence of ${\bf P1}$ and ${\bf P3}$ can be proved according to \cite{shiIterativelyWeightedMMSE2011,huIterativeAlgorithmInduced2021}. The objective function of ${\bf P3}$ is convex with respect to each variable, which simplifies the following derivation.

\subsection{Active Beamforming Design}

By iteratively optimizing one variable and fixing other variables, the closed-form solution of $u_k$, $\lambda_k$ and $\mathbf{w}_k$ can be easily derived as
\begin{align}
  u_k^{\rm opt} & = \left(\mathbf{h}_k\sum_{i = 1}^K\mathbf{w}_i \mathbf{w}_i^H \mathbf{h}_k^H + \frac{\sigma_k^2}{P_\mathrm{T}} \sum_{i = 1}^K\mathbf{w}_i^H \mathbf{w}_i \right)^{-1}\mathbf{h}_k \mathbf{w}_k, \label{equ_u_k}\\
  \lambda_k^{\rm opt} & = \left(1 - u_k^* \mathbf{h}_k \mathbf{w}_k\right)^{-1}, \label{equ_lambda_k}\\
  \mathbf{w}_k^{\rm opt} & = \alpha_k u_k \lambda_k \left(\sum_{i = 1}^K \alpha_i |u_i|^2 \lambda_i \left(\frac{\sigma_i^2}{P_\mathrm{T}} + \mathbf{h}_i^H\mathbf{h}_i\right)\right)^{-1}\mathbf{h}_k^H.\label{equ_w_k}
\end{align}
Thus we obtain the optimal active beamforming (\ref{equ_w_k}) when other optimization variables are fixed.
\subsection{Passive Beamforming Design}
\par
By fixing $\{u_k, \lambda_k, \mathbf{w}_k\}$, ${\bf P3}$ can be simplified into an optimization related to the passive beamforming vector as
\begin{equation}
  \begin{aligned}
  ({\bf P4}): \quad \min _{\boldsymbol{\theta}} \,\, &\boldsymbol{\theta}^H \mathbf{A} \boldsymbol{\theta} + \boldsymbol{\beta}^H \boldsymbol{\theta} + \boldsymbol{\theta}^H \boldsymbol{\beta} \\
  \text{s.t.}\,\, &\theta_{n} \in \mathcal{F}, \quad \forall n=1, \ldots, N,
  \end{aligned}
\end{equation}
where
\begin{align}
\mathbf{A} & = \sum_{k = 1}^{K} \alpha_k \lambda_k |u_k|^2 \mathbf{H}_{\mathrm{r}, k}\sum_{i=1}^K \mathbf{w}_i \mathbf{w}_i^H \mathbf{H}_{\mathrm{r}, k}^H,\\
\boldsymbol{\beta} & = \sum_{k = 1}^{K} - \alpha_k \lambda_k u_k^* \mathbf{H}_{\mathrm{r}, k} \mathbf{w}_k + \alpha_k \lambda_k |u_k|^2 \mathbf{H}_{\mathrm{r}, k} \sum_{i=1}^K \mathbf{w}_i \mathbf{w}_i^H \mathbf{h}_{\mathrm{d}, k}^H.
\end{align}
Although the objective function is convex to $\boldsymbol{\theta}$, the unit module constraints are non-convex. The proplem is still hard to be solved directly. Inspired by \cite{wuIntelligentReflectingSurface2019}, we introduce an auxiliary variables $z$, and denote $\boldsymbol{x} = [\boldsymbol{\theta}^T, z]^T \in \mathbb{C}^{(N+1)\times 1}$. Therefore, ${\bf P4}$ can be equivalently transformed into
\begin{equation}
  \begin{aligned}
  ({\bf P5}): \quad\max_{\boldsymbol{x}} \,\, &\boldsymbol{x}^H \mathbf{B} \boldsymbol{x}\\
  \text{s.t.}\,\, &\, x_{n} \in \mathcal{F}, \quad \forall n=1, \ldots, N+1,
  \end{aligned}
\end{equation}
where $x_n$ is the $n$-th element of $\boldsymbol{x}$ and $\mathbf{B}$ can be denoted as
\begin{equation}
  \mathbf{B} =\left[\begin{array}{cc}
  -\mathbf{A} & -\boldsymbol{\beta} \\
  -\boldsymbol{\beta}^H & 0
  \end{array}\right].
  \label{equ_B}
\end{equation}
The optimization problem ${\bf P5}$ is a unimodular quadratic program and can be solved by the PI algorithm \cite{soltanalianDesigningUnimodularCodes2014}, which uses the iterative solution as
\begin{equation}
  \boldsymbol{x}^{(p+1)} = e^{\jmath \operatorname{angle}(\mathbf{R}\boldsymbol{x}^{(p)})}.
  \label{equ_PI}
\end{equation}
Here, $\boldsymbol{x}^{(p)}$ is the value of $\boldsymbol{x}$ in the $p$-th iteration. And $\mathbf{R} = \mathbf{B} + \gamma \mathbf{I}$, where $\gamma\mathbf{I}$ is introduced to ensure that $\mathbf{R}$ is positive definite. This guarantees the convergence of the iterative algorithm, as shown in Lemma \ref{theo_1}.
\begin{lemma}
  The PI algorithm is guaranteed to converge to at least a local optimum of ${\bf P5}$ when $\mathbf{R}$ is positive-defined.
  \label{theo_1}
\end{lemma}
\begin{proof}
  \par
  The proof is provided in Appendix \ref{Appendix_A}.
\end{proof}
We set $\gamma = \|\mathbf{B}\|$ in the following sections\footnote{Setting the value of $\gamma$ slightly larger than the minimum eigenvalue of $\mathbf{B}$ is sufficient to ensure the positive definiteness of $\mathbf{R}$. However, calculating eigenvalues introduces higher computational complexity compared to computing Frobenius norms. As a result, we opt to set $\gamma = \|\mathbf{B}\|$ in order to mitigate the computational burden.}. After the iteration converges, the passive beamforming vector can be derived as $\boldsymbol{\theta} = x_{N+1}^* \times \boldsymbol{x}[1:N]^T$, which is the multiple of conjugate of the last term of $\boldsymbol{x}$ and the first $N$ terms of $\boldsymbol{x}$.

\subsection{Proposed Joint Beamforming Algorithm}
\begin{algorithm}[tb]
  \SetKwInOut{Input}{input}\SetKwInOut{Output}{output}
  \Input{$\mathbf{G}$, $\mathbf{h}_{\mathrm{d}, k}$, $\mathbf{h}_{\mathrm{r}, k}$, $\alpha_k$, for $k = 1,2,\ldots, K$}
  \Output{$\boldsymbol{\theta}$, $\mathbf{W}$}

  randomly initialize $\boldsymbol{\theta}$ \label{alg1_init1}\;
  initialize $\mathbf{W}$ with ZF beamforming \label{alg1_init2}\;
  calculate $u_k$ and $\lambda_k$ according to (\ref{equ_u_k}) and (\ref{equ_lambda_k}) \label{alg1_init3}\;
  update $\mathbf{W}$ according to (\ref{equ_w_k}) \label{alg1_init4}\;
  set $ i = 1$\;
  \While{\textbf{$i \leq I_O$} \label{alg1_iter0}}{
    calculate $u_k$ and $\lambda_k$ according to (\ref{equ_u_k}) and (\ref{equ_lambda_k}) \label{alg1_iter1}\;
    update $\boldsymbol{\theta}$ according to (\ref{equ_PI}) until the objective function of ${\bf P5}$ converges \label{alg1_iter2}\;
    update $\mathbf{W}$ according to (\ref{equ_w_k}) and scale $\mathbf{W}$ according to the maximum transmit power constraint \label{alg1_iter3}\;
    $i = i+1$\;
  }\label{alg1_iter4}
  map $\theta_n$ to the nearest element in $\mathcal{F}$ \label{alg1_quantization}\;
  calculate $u_k$ and $\lambda_k$ according to (\ref{equ_u_k}) and (\ref{equ_lambda_k}) \label{alg1_updateW_1}\;
  update $\mathbf{W}$ according to (\ref{equ_w_k}) and scale $\mathbf{W}$ according to the maximum transmit power constraint \label{alg1_updateW_2}\;
  \caption{WMMSE-PI Algorithm for Joint Beamforming with Perfect CSI}
  \label{alg_WMMSE_PI_Perfect}
\end{algorithm}

The WMMSE-PI algorithm is formulated based on the aforementioned derivations and is outlined in Algorithm \ref{alg_WMMSE_PI_Perfect}. Initialization occurs in Steps \ref{alg1_init1} to \ref{alg1_init4}, with active beamforming vectors initialized using the WMMSE algorithm (Steps \ref{alg1_init3} and \ref{alg1_init4}), and the passive beamforming vector initialized randomly. The primary section of the algorithm comprises two levels of iteration. The outer $I_O$ iteration updates the active and passive beamforming vectors sequentially, covering Steps \ref{alg1_iter0} to \ref{alg1_iter4} of Algorithm \ref{alg_WMMSE_PI_Perfect}. Within Step \ref{alg1_iter2}, the PI algorithm is employed to iteratively refine the passive beamforming vector---an operation that constitutes the inner iteration. To adhere to the discrete phase shift constraint, the output passive beamforming vector is quantized as outlined in Step \ref{alg1_quantization}. Steps \ref{alg1_updateW_1} and \ref{alg1_updateW_2} facilitate a single update of the active beamforming vector through the WMMSE algorithm to alleviate potential degradation attributed to quantization. The quantization process is performed solely once after the last iteration, preserving the algorithm's convergence. To prevent confusion, we employ superscript ``$p$'' to denote the value of the $p$-th inner iteration variable, and superscript ``$i$'' is used to represent the value of the $i$-th outer iteration variable.

Notably, Step \ref{alg1_iter2} is derived from maximizing the objective function of ${\bf P5}$, which is equivalent to minimizing the objective function of ${\bf P3}$. Considering that Steps \ref{alg1_iter1} and \ref{alg1_iter3} are also designed to minimize the objective function of ${\bf P3}$, the objective function of ${\bf P3}$ continually decreases throughout the two-layer iterations. As a result, Algorithm \ref{alg_WMMSE_PI_Perfect} ensures the convergence of the objective function of ${\bf P3}$, thus confirming the convergence of the equivalent ${\bf P1}$.

\section{Model-Driven DL for Joint Beamforming}
\label{Sec_Model_Driven}
\par
This section introduces a model-driven DL approach designed to mitigate the number of iterations required by the WMMSE-PI algorithm, thereby reducing computational complexity. This approach leverages DL to learn crucial parameters associated with convergence speed, and incorporates a GNN to enhance the initialization of the active beamforming vector beyond what the WMMSE algorithm computes. To address the constraint of discrete phase shifts in passive beamforming, an approximate quantization function is integrated to enable backpropagation during the DL training process. After an in-depth presentation of the model-driven DL, a comprehensive comparison and analysis of the complexity among the proposed techniques and other existing algorithms is provided.
\subsection{Convergence Acceleration}
\par
The convergence of the WMMSE-PI algorithm is guaranteed, but the convergence speed is dependent on the value of $\gamma$. For example, when $\gamma \rightarrow +\infty $, $\boldsymbol{x}^{(p+1)} = e^{\jmath \operatorname{angle}\left(\left(\mathbf{B}+\gamma \mathbf{I}\right)\boldsymbol{x}^{(p)}\right)} \rightarrow e^{\jmath \operatorname{angle}\left(\gamma \boldsymbol{x}^{(p)}\right)} = \boldsymbol{x}^{(p)}$, causing the algorithm to slow down. However, if $\gamma$ is too small, the positive definiteness of $\mathbf{R}$ is not guaranteed, and the algorithm's convergence may be compromised, resulting in degraded performance. To expedite the convergence, we design a model-driven DL by selecting $\gamma$ as a trainable variable.
\par
To update the passive beamforming, we iterate with (\ref{equ_PI}) until the objective function of ${\bf P5}$ converges in Algorithm \ref{alg_WMMSE_PI_Perfect}. The complexity is linear with the number of inner iterations. Hence, we set the number of inner iterations as 1 after the introduction of trainable variables. This modification also reduces the complexity of one outer iteration. The WMMSE-PI algorithm with $\gamma$ as trainable variables is termed as ``WMMSE-PINet''. The $\gamma$ is set as an  independent training variable in every outer iteration, and $\gamma^{(i)}$ stands for the value in the $i$-th outer iteration.

\subsection{Initialization with GNN}
\par
The required number of iterations not only depend on the convergence speed but also the initialization. In the proposed WMMSE-PI algorithm, the active beamforming vector is initialized by the ZF precoding and then updated with the WMMSE algorithm once, which may be far away from the global optimum and result in more iterations. Since the update of passive beamforming also depends on active beamforming, a better initialization of active beamforming will also lead to faster convergence of the algorithm. Therefore, we enhance the initialization of active beamforming by introducing the DL and considering the structure of optimal active beamforming vectors.
\par
When the variance of the received signal noise is equal, that is, $\sigma_1^2 = \sigma_2^2 = \cdots = \sigma_K^2 = \sigma^2$ \footnote{The assumption is reasonable when the users are in a similar environment and the user device type are the same. Because the variance of the received signal noise is mainly determined by natural noise and artificial noise.}, optimal active beamforming vector in (\ref{equ_w_k}) can be written as \cite{xiaDeepLearningFramework2020}
\begin{equation}
  \mathbf{w}_k = \sqrt{p_k} \, \frac{\left(\sigma^2\mathbf{I} + \sum_{i = 1}^K \zeta_k \mathbf{h}_i^H\mathbf{h}_i \right)^{-1}\mathbf{h}_k^H}{\Big\|\left(\sigma^2\mathbf{I} + \sum_{i = 1}^K \zeta_k \mathbf{h}_i^H\mathbf{h}_i \right)^{-1}\mathbf{h}_k^H \Big\|},
  \label{equ_w_structure}
\end{equation}
where $p_k$ and $\zeta_k$ are positive parameters and satisfy $\sum_{k = 1}^K p_k = \sum_{k = 1}^K \zeta_k = P_T$. If $p_k$ and $\zeta_k$ are obtained, the optimal active beamforming vectors can be easily calculated without iterations. We utilize DL to learn the values of $p_k$ and $\zeta_k$, which provides better interpretability and performance compared to directly learning the active beamforming vectors.
\par
The interactions in wireless communication networks can be well represented by the GNN \cite{shenGraphNeuralNetworks2021, jiangLearningReflectBeamform2021}. Specifically, the transmission signal and interference between users can be well represented by the combination and aggregation function in GNN. Additionally, the objective function of ${\bf P1}$ exhibits the properties of permutation invariance which implies that the order of the user index does not affect the WSR of the system. However, the neural network should be permutation equivariant, meaning that if the input vectors to the neural network are permuted, the output vectors should be permuted in the same way. Compared with the traditional fully connected neural network, GNN better embeds the properties of permutation invariance and permutation equivariance \cite{shenGraphNeuralNetworks2021, shenGraphNeuralNetworks2022}. Therefore, we adopt GNN as the network structure and design the input and output of GNN. GNN consists of the input layer, update layer, and output layer, which we describe in detail later. To avoid any confusion, we employ the superscript ``$g$'' to denote the value associated with the $g$-th layer variables.

\begin{figure*}[htbp!]
  \centering
  \includegraphics[width=0.9\textwidth]{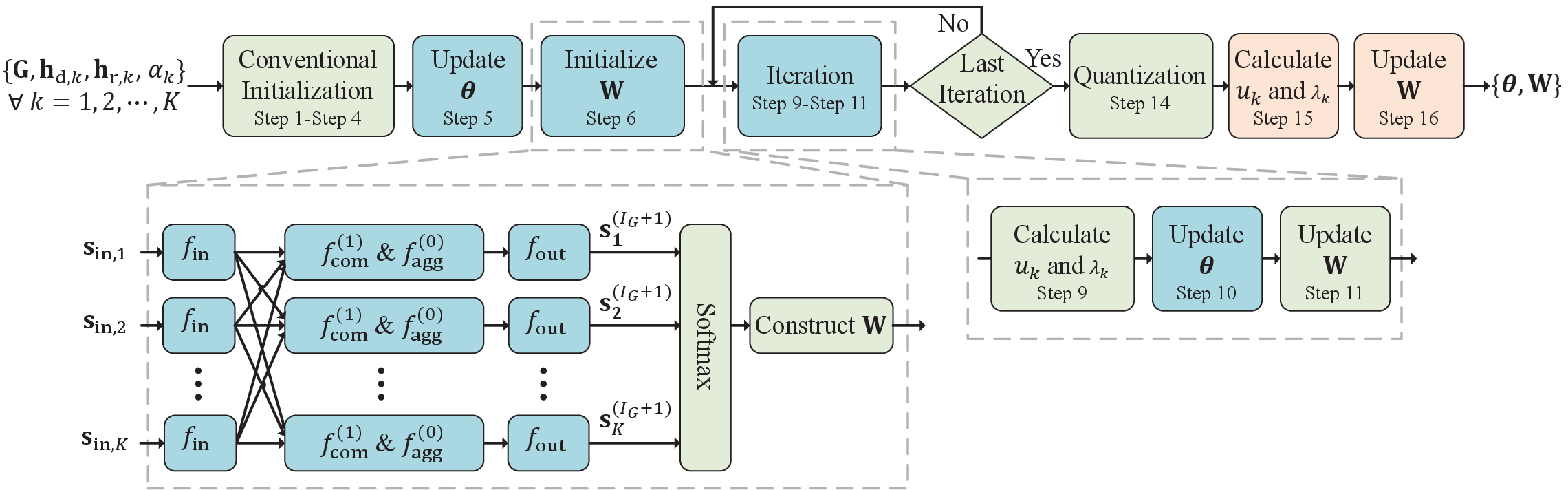}
  \caption{Network structure of WMMSE-PINet$^+$.}
  \label{fig_network_structure}
\end{figure*}

\begin{itemize}
  \item \textbf{Input layer:} The input layer converts input vectors into information flowing inside the network, which is denoted as
  \begin{equation}
    \mathbf{s}_k^{(0)} = f_{\mathrm{in}}(\mathbf{s}_{\mathrm{in}, k}),
  \end{equation}
  where $f_{\mathrm{in}}: \mathbb{R}^{4M+6} \rightarrow \mathbb{R}^{6M}$ is one fully-connected layer with ReLU as the activation function. The variables, $\tilde{\mathbf{s}}_{\mathrm{in}, k} = \{\mathbf{h}_{\mathrm{d}, k}^T, (\boldsymbol{\theta}^H\mathbf{H}_{\mathrm{r}, k})^T, u_k, \lambda_k, \alpha_k\}$, are contained in (\ref{equ_w_k}), which determins the optimal active beamforming vectors and thus the values of $p_k$ and $\zeta_k$. We stack the real and imaginary parts of $\tilde{\mathbf{s}}_{\mathrm{in}, k}$ as the input of neural network $\mathbf{s}_{\mathrm{in}, k} = \{\Re{(\tilde{\mathbf{s}}_{\mathrm{in}, k})}, \Im{(\tilde{\mathbf{s}}_{\mathrm{in}, k})} \} \in \mathbb{R}^{4M+6}$, because the complex values are not well-supported in common DL framework.
  \item \textbf{Update layer:} We denote $\mathbf{s}_k^{(g)}$ as the information of user $k$ in the $g$-th layer. The update of $\mathbf{s}_k^{(g)}$ should take $\mathbf{s}_k^{(g-1)}$ and the aggregated information of its neighboring nodes into consideration, $\{\mathbf{s}_j^{(g-1)}\}_{j\in \mathcal{N}_k}$, where $\mathcal{N}_k$ are the users without user $k$. The update layer is denoted as
  \begin{equation}
    \mathbf{s}_k^{(g)} = f_{\mathrm{com}}^{(g)}\left(\mathbf{s}_k^{(g-1)}, f_{\mathrm{agg}}^{(g-1)}\left(\psi\left(\{\mathbf{s}_j^{(g-1)}\}_{j\in \mathcal{N}_k}\right)\right)\right),
  \end{equation}
  where $f_{\mathrm{agg}}^{(g-1)}:\mathbb{R}^{6M}\rightarrow\mathbb{R}^{3M}$ represents the aggregation function, which is one fully-connected layer with ReLU as the activation function. To achieve the permutation invariant property, $\psi(\cdot)$ returns the result of element-wise max-pooling. The combination function is denoted as $f_{\mathrm{com}}^{(g-1)}:\mathbb{R}^{9M}\rightarrow\mathbb{R}^{3M}$, which is one fully-connected layer with ReLU as the activation function.
  \item \textbf{Output layer:} After the $I_G$ update layers, the information flowing inside the network should be converted to the output vectors used for constructing the active beamforming vectors, that is, $\mathbf{s}_{\mathrm{out},k} = \{p_k, \zeta_k\} \in \mathbb{R}^{2}$. The output layer is denoted as
  \begin{equation}
    \mathbf{s}_k^{(I_G+1)} = f_{\mathrm{out}}(\mathbf{s}_{k}^{(I_G)}),
  \end{equation}
  where $f_{\mathrm{out}}: \mathbb{R}^{3M} \rightarrow \mathbb{R}^{2}$ is one fully-connected layer without activation function. To satisfy the constraint, $\sum_{k = 1}^K p_k = \sum_{k = 1}^K \zeta_k = P_T$, the output vectors pass through a softmax layer, denoted as
  \begin{equation}
    \mathbf{s}_{\mathrm{out},k} = \frac{e^{\mathbf{s}_k^{(I_G+1)}}}{\sum_{k=1}^K e^{\mathbf{s}_k^{(I_G+1)}}} P_T,
  \end{equation}
  where the element-wise division is performed. Therefore, the active beamforming vectors can be constructed based on the output of GNN.
\end{itemize}

After introducing the details of GNN, we set the number of update layers as $I_G=1$, which achieve a better trade-off between performance and initialization complexity. WMMSE-PINet combined with GNN initialization is denoted as ``WMMSE-PINet$^+$''. Although GNN can be considered as a data-driven DL method, WMMSE-PINet$^+$ exhibits strong robustness due to the integrated expert knowledge, which will be evaluated in Section. \ref{sec_robustness}.

\subsection{Approximate Quantization Function}
\par
When considering the discrete phase shift of RIS, it is necessary to quantize $\boldsymbol{\theta}$ after the iteration of WMMSE-PINet$^+$. However, using traditional quantization functions may result in gradients of 0 or infinity during backpropagation, which halts the neural network training. A feasible solution is to use an approximate quantization function, as proposed in \cite{xuRobustDeepLearningBased2022}, which can be written as
\begin{equation} 
  Q(\boldsymbol{\theta}) = \frac{\pi}{2^B} \sum_{l=1}^{2^B}\operatorname{tanh}\left(\eta\left(\boldsymbol{\theta} - \frac{2\pi l}{2^B} \right)\right) + 1,
  \label{equ_quantization}
\end{equation}
where $B$ is the number of phase shifter bits, $\eta$ represents the approximate level. Higher values of $\eta$ lead to better approximation. The function $Q(\boldsymbol{\theta})$ ensures gradients are well-defined at all points. During training, the traditional quantization in Step \ref{alg1_quantization} of the Algorithm \ref{alg_WMMSE_PI_Perfect} is replaced by the approximate quantization function. The traditional quantization function is still used during testing and deploying.

\subsection{Model-driven DL WMMSE-PINet$^+$}
\begin{algorithm}[tb]
  \SetKwInOut{Input}{input}\SetKwInOut{Output}{output}
  \Input{$\mathbf{G}$, $\mathbf{h}_{\mathrm{d}, k}$, $\mathbf{h}_{\mathrm{r}, k}$, $\alpha_k$, for $k = 1,2,\ldots, K$}
  \Output{$\boldsymbol{\theta}$, $\mathbf{w}_{k}$, for $k = 1,2,\ldots, K$}

  randomly initialize $\boldsymbol{\theta}$ \label{alg2_init1}\;
  initialize $\mathbf{W}$ with ZF beamforming\label{alg2_init2}\;
  calculate $u_k$ and $\lambda_k$ according to (\ref{equ_u_k}) and (\ref{equ_lambda_k})\label{alg2_init3}\;
  update $\mathbf{W}$ according to (\ref{equ_w_k}) \label{alg2_init4}\;
  update $\boldsymbol{\theta}$ according to (\ref{equ_PI}) and trainable variable $\gamma^{(0)}$\label{alg2_init5}\;
  initialize $\mathbf{W}$ with GNN \label{alg2_init6}\;
  set $ i = 1$\;
  \While{\textbf{$i \leq I_O$} \label{alg2_iter0}}{
    calculate $u_k$ and $\lambda_k$ according to (\ref{equ_u_k}) and (\ref{equ_lambda_k})\;
    update $\boldsymbol{\theta}$ according to (\ref{equ_PI}) and trainable variable $\gamma^{(i)}$\;
    update $\mathbf{W}$ according to (\ref{equ_w_k}) and scale $\mathbf{W}$ according to the maximum transmit power constraint\;
    $i = i+1$\;
  }\label{alg2_iter4}
  map $\theta_n$ to the nearest element in $\mathcal{F}$ (using an approximate quantization function, $Q(\boldsymbol{\theta})$, during training, and use a traditional quantization function during testing)\;
  calculate $u_k$ and $\lambda_k$ according to (\ref{equ_u_k}) and (\ref{equ_lambda_k})\;
  update $\mathbf{W}$ according to (\ref{equ_w_k}) and scale $\mathbf{W}$ according to the maximum transmit power constraint\;
  \caption{WMMSE-PINet$^+$ for Joint Beamforming with Perfect CSI}
  \label{alg_WMMSE_PINet_Perfect}
\end{algorithm}

\begin{table*}[tbp]
  \caption{Complexity Analysis}
  \label{tab_complexity}
    \centering
    \begin{tabular}{cccc}
      \toprule
      & Computational Complexity & Number of Iterations Required & Runtime (ms)  \\
      \midrule
      FP-ADMM \cite{guoWeightedSumRateMaximization2019} & $\mathcal{O}\left(I_O\left(I_A(KNM+KM^2)+N^3\right)\right)$ & 28 & 216.4485      \\
      WMMSE-MO \cite{guoWeightedSumRateMaximization2020}& $\mathcal{O}\left(I_O \left(I_{\lambda}I_{W}KM^3 + I_R K^2N^2\right)\right)$ & 22 & 260.788 \\
      BCD \cite{guoWeightedSumRateMaximization2020}& $\mathcal{O}\left(I_O\left(I_A(KNM+KM^2)+K^2N^2\right)\right)$ & 50 & 105.8293
      \\
      WMMSE-PI & $\mathcal{O}\left(I_O\left(KM^3+K^2N^2 + I_P N^2\right)\right)$ & 38
      & 87.5765      \\
      WMMSE-PINet & \textbf{$\mathcal{O}\left(I_O\left(K M^3+ K^2N^2\right)\right)$} & 10 & 6.445 \\
      WMMSE-PINet$^+$ & \textbf{$\mathcal{O}\left(KM^2 + I_O\left(K M^3+ K^2N^2\right)\right)$} & 3 & \textbf{2.564} \\
      \bottomrule
    \end{tabular}
\end{table*}

WMMSE-PINet$^+$ is presented as Algorithm \ref{alg_WMMSE_PINet_Perfect} and the structure is depicted in Fig.~\ref{fig_network_structure}. Steps \ref{alg2_init1} to \ref{alg2_init4} are called as ``Conventional Initialization''. Compared to WMMSE-PINet, Steps \ref{alg2_init5} and \ref{alg2_init6} are introduced to improve the initialization of passive and active beamforming, respectively. The passive beamforming is initialized randomly in Step \ref{alg2_init1}, and the elements of $\boldsymbol{\theta}^H\mathbf{H}_{\mathrm{r}, k}$ would also be random if the Step \ref{alg2_init5} is not added. $\gamma^{(0)}$ represents the value of $\gamma$ used in Step \ref{alg2_init5}.

The block containing trainable variables in Fig.~\ref{fig_network_structure} is colored blue. Because the values of $\gamma^{(i)}$ in every outer iteration are set as unshared, the total trainable variables are $\boldsymbol{\Gamma} = \{\gamma^{(0)}, \gamma^{(1)}, \ldots, \gamma^{(I_O)}, f_{\mathrm{in}}, f_{\mathrm{agg}}^{(0)}, f_{\mathrm{com}}^{(1)}, f_{\mathrm{out}}\}$. Each of the ${\{f_{\mathrm{in}}, f_{\mathrm{agg}}^{(0)}, f_{\mathrm{com}}^{(1)}, f_{\mathrm{out}}\}}$ contains a weighted matrix and bias. Therefore, the total number of trainable variables is $I_O + 69M^2 + 54M + 3$, which is typically smaller than that of traditional data-driven neural networks. WMMSE-PINet and WMMSE-PINet$^+$ are trained using the WSR as the loss function. During the training process, the approximate quantization function is utilized instead of the traditional one to avoid the gradient becoming 0 or infinity.

\subsection{Complexity Analysis}
\par
We compare the complexity of the proposed algorithms with other joint beamforming algorithms, including the fractional programing-alternating direction method of multipliers (FP-ADMM) algorithm \cite{guoWeightedSumRateMaximization2019}, the WMMSE-manifold optimization (WMMSE-MO) algorithm \cite{guoWeightedSumRateMaximization2020}, and the BCD algorithm \cite{guoWeightedSumRateMaximization2020}.
The complexities and runtime of the algorithms are listed in Table \ref{tab_complexity}, where $\{I_A, I_{\lambda}, I_{W}, I_R\}$ represent the number of inner iterations in the algorithms as noted in their respective papers. $I_P$ is the number of inner iterations of the PI algorithm, and $I_O$ is the number of outer iterations. In RIS-assisted systems, the number of reflective elements is usually large, so the term that mainly contributes to the computational complexity is the one which contains $N$. From Table \ref{tab_complexity}, it can be observed that except for the FP-ADMM algorithm, other algorithms only include the square term of $N$. In our proposed three algorithms, the computational complexity mainly arises from the calculation of $\mathbf{A}$ when updating the passive beamforming vector, which is $\mathcal{O}\left(K^2 N^2\right)$. In WMMSE-PINet$^+$, the complexity attributed to the initialization is $\mathcal{O}\left(K M^2\right)$, which can be ignored compared to the complexity of subsequent iterations.
\par
To compare the computational complexity of different algorithms more intuitively, we choose the performance of WMMSE-PINet$^+$ with $I_O = 3$ iterations as a benchmark and compare the number of iterations and runtime of other algorithms required to achieve the same performance, assuming $M = 8$, $K = 4$, $N = 100$, and $P_\mathrm{T} = 10$ dBm. All algorithms are implemented on a PC with Windows 10 and Intel Core i7-10700 CPU @ 2.90GHz. The BCD algorithm have the lowest runtime among the state-of-the-art algorithms. Although the complexity of the WMMSE-MO algorithm does not contain the $N^3$ term, its runtime is longer than the FP-ADMM algorithm due to large coefficient of the $KM^3$ term. Compared with the BCD algorithm, the proposed WMMSE-PI algorithm reduces the runtime to achieve the same performance to about 90\%. Because the numbers of outer and inner iteration are both reduced by WMMSE-PINet, the runtime to achieve the same performance is about 6\% of the BCD algorithm, which is significantly reduced. As for WMMSE-PINet$^+$, although the GNN is introduced as the initialization, the runtime is less than 3\% of the BCD algorithm due to the fewer required number of iterations.

\section{Extended WMMSE-PINet-ImCSI for Imperfect CSI Scenarios}
\label{Sec_ImCSI}

This section introduces an extension of the proposed WMMSE-PINet$^+$ algorithm, designed to address the two-timescale optimization problem, denoted as ${\bf P2}$. Motivated by the online stochastic optimization framework outlined in \cite{liuOnlineSuccessiveConvex2018}, we resolve the short-term subproblem by employing the current combined channel. Meanwhile, the long-term master problem is tackled by employing sampled cascaded channels.

\subsection{Solve the Short-term Subproblem}
\par
The subproblem for optimizing the active beamforming vectors can be formulated as follows:
\begin{equation}
  \begin{aligned}
  ({\bf P6}):\quad \max _{\mathbf{W}^{[\xi]}} \,\, &\sum_{k=1}^{K} \alpha_{k} \mathcal{R}_k^{[\xi]} \\
  \text{s.t.}\,\, &\sum_{k=1}^K\left\|\mathbf{w}_k^{[\xi]}\right\|^2 \leq P_{\mathrm{T}}.
  \end{aligned}
\end{equation}
This subproblem can be transformed into the form of ${\bf P3}$ and solved using the WMMSE algorithm, providing a local optimum for the inner short-term optimization. After applying the WMMSE algorithm, the outer long-term optimization can be simplified as:
\begin{equation}
  \begin{aligned}
    ({\bf P7}): \quad \max_{\boldsymbol{x}} \,\, &\mathbb{E}_\xi [\boldsymbol{x}^H \mathbf{B}^{[\xi]} \boldsymbol{x}]\\
  \text{s.t.}\,\, &\, x_{n} \in \mathcal{F}, \quad \forall n=1, \ldots, N+1,
  \end{aligned}
\end{equation}
Here, $\mathbf{B}^{[\xi]}$ is calculated by substituting channel samples $\{\mathbf{h}_{\mathrm{d},k}^{[\xi]}, \mathbf{H}_{\mathrm{r},k}^{[\xi]}\}$ into (\ref{equ_u_k}), (\ref{equ_lambda_k}), (\ref{equ_w_k}), and (\ref{equ_B}). This calculation occurs only once and is used to obtain the stationary solution for both the short-term subproblem and the long-term master problem.

\subsection{Solve the Long-term Matser Problem} 

The expectation operator in ${\bf P6}$ poses a challenge to solve the problem. One classical approach is to use the sample average approximation method, which samples a new group of channels in each outer iteration and then uses the average of the samples to approximate the expectation. However, it requires significant memory overhead to store $\mathbf{B}^{[\xi]}, \, \xi = 1, \ldots, i$, in the $i$-th outer iteration. Inspired by the SSCA algorithm, a recursive approximation can be applied as:
\begin{equation}
  \begin{aligned}
  ({\bf P8}): \quad \max_{\boldsymbol{x}} \,\, &\boldsymbol{x}^H \tilde{\mathbf{B}}^{(i)} \boldsymbol{x}\\
  \text { s.t.}\,\, & \left|x_{n}\right|=1, \quad \forall n=1, \ldots, N+1.
  \end{aligned}
\end{equation}
In this formulation, $\tilde{\mathbf{B}}^{(i)}$ is calculated using a recursive relation:
\begin{equation}
  \tilde{\mathbf{B}}^{(i)} = (1 - \rho^{(i)})\tilde{\mathbf{B}}^{(i-1)} + \rho^{(i)} \mathbf{B}^{(i)},
  \label{equ_recursive_B}
\end{equation}
where $\tilde{\mathbf{B}}^{(0)}$ is initialized to $ \mathbf{0}$. In each iteration, a new group of channels is sampled and indexed as $i$, which is then used to calculate the corresponding $\mathbf{B}^{(i)}$ according to (\ref{equ_B}). Both optimization problems ${\bf P8}$ and ${\bf P5}$ share the same form, allowing the PI algorithm to be applied for solving ${\bf P8}$. The PI algorithm with trainable variable $\gamma^{(i)}$ is executed, resulting in a solution denoted as $\bar{\boldsymbol{\theta}}^{(i)}$. Finally, a recursive update is performed for the passive beamforming vector:
\begin{equation}
  \boldsymbol{\theta}^{(i)} = (1-\delta^{(i)})\boldsymbol{\theta}^{(i-1)} + \delta^{(i)} \bar{\boldsymbol{\theta}}^{(i)}.
  \label{equ_recursive_theta}
\end{equation}
Here, $\boldsymbol{\theta}^{(0)}$ is obtained during the initialization process. The values of $\{\rho^{(i)}, \delta^{(i)}\}$  determine the past samples on the current iteration, with smaller $\{\rho^{(i)}, \delta^{(i)}\}$ meaning the past samples have greater impact in the $i$-th iteration. Suitable values of $\{\rho^{(i)}, \delta^{(i)}\}$ can lead to improved algorithm performance. Consequently, we consider $\{\rho^{(i)}, \delta^{(i)}\}$ as trainable variables, which are optimized along with the other parameters of WMMSE-PINet$^+$ to achieve better performance.

\begin{algorithm}[tb]
  \SetKwInOut{Input}{input}\SetKwInOut{Output}{output}
  \Input{$\mathbf{h}_{\mathrm{d}, k}^{[i]}$, $\mathbf{G}^{[i]}$, $\mathbf{h}_{\mathrm{r}, k}^{[i]}$, $\alpha_k$, for $k = 1,2,\ldots, K \text{ and } \xi = 0,1,\ldots, I_O$}
  \Output{$\boldsymbol{\theta}$}

  sample a group of channels $\{\mathbf{h}_{\mathrm{d},k}^{[0]}, \mathbf{G}^{[0]}, \mathbf{h}_{\mathrm{r},k}^{[0]}\}$ \label{alg3_init1}\;
  randomly initialize $\boldsymbol{\theta}$\;
  initialize $\mathbf{W}$ with ZF beamforming\;
  calculate $u_k$ and $\lambda_k$ according to (\ref{equ_u_k}) and (\ref{equ_lambda_k})\;
  update $\mathbf{W}$ according to (\ref{equ_w_k})\;
  update $\boldsymbol{\theta}$ according to (\ref{equ_PI}) and trainable variable $\gamma^{(0)}$, denote the result as $\boldsymbol{\theta}^{(0)}$\;
  initialize $\mathbf{W}$ with GNN\;
  set $ i = 1$\;
  \While{\textbf{$i \leq I_O$} \label{alg2_iter0}}{
    sample a group of channels $\{\mathbf{h}_{\mathrm{d},k}^{[i]}, \mathbf{G}^{[i]}, \mathbf{h}_{\mathrm{r},k}^{[i]}\}$ \label{alg3_update1}\;
    calculate $u_k$ and $\lambda_k$ according to (\ref{equ_u_k}) and (\ref{equ_lambda_k})\label{alg3_update2}\;
    update $\mathbf{W}$ according to (\ref{equ_w_k}) and scale $\mathbf{W}$ according to the maximum transmit power constraint\label{alg3_update3}\;
    construct $\mathbf{B}^{(i)}$ according to (\ref{equ_B}) \label{alg3_update4}\;
    update $\tilde{\mathbf{B}}^{(i)}$ according to (\ref{equ_recursive_B}) and trainable variable $\rho^{(i)}$\label{alg3_update5}\;
    update $\bar{\boldsymbol{\theta}}^{(i)}$ according to (\ref{equ_PI}) and trainable variable $\gamma^{(i)}$\label{alg3_update6}\;
    update $\boldsymbol{\theta}^{(i)}$ according to (\ref{equ_recursive_theta}) and trainable variable $\delta^{(i)}$\label{alg3_update7}\;
    $i = i+1$\;
  }
  map $\theta_n$ to the nearest element in $\mathcal{F}$ (using an approximate quantization function, $Q(\boldsymbol{\theta})$, during training, and use a traditional quantization function during testing)\;
  \caption{WMMSE-PINet-ImCSI for Joint Beamforming with Imperfect CSI}
  \label{alg_WMMSE_PINet_Imperfect}
\end{algorithm}
\subsection{Model-driven DL WMMSE-PINet-ImCSI}
Drawing inspiration from the SSCA algorithm, we propose an extension to WMMSE-PINet$^+$ to handle the imperfect CSI scenario, referred to as ``WMMSE-PINet-ImCSI,'' and presented in Algorithm \ref{alg_WMMSE_PINet_Imperfect}. The initialization of WMMSE-PINet-ImCSI is similar to that of WMMSE-PINet$^+$, involving the sampling of a group of channels in Step \ref{alg3_init1}. However, in WMMSE-PINet-ImCSI, a new group of channels is sampled in Step \ref{alg3_update1} of each iteration to account for the presence of imperfect CSI. In Step \ref{alg3_update3}, the active beamforming vectors are updated before the passive beamforming because solving the short-term subproblem is a prerequisite for tackling the long-term problem. Furthermore, Steps \ref{alg3_update4} to \ref{alg3_update7} involve a different approach for updating the passive beamforming vectors compared to WMMSE-PINet$^+$ due to the imperfect estimation of the cascaded channels.

The introduction of two additional variables, $\{\rho^{(i)}, \delta^{(i)}\}$, in each iteration requires the expansion of the trainable variables of WMMSE-PINet-ImCSI, denoted as $\boldsymbol{\Gamma}' = \{\gamma^{(0)}, \gamma^{(1)}, \ldots, \gamma^{(I_O)}, \rho^{(1)}, \rho^{(2)}, \ldots, \rho^{(I_O)}, \delta^{(1)}, \delta^{(2)}, \ldots, $\\$\delta^{(I_O)},  f_{\mathrm{in}}, f_{\mathrm{agg}}^{(0)}, f_{\mathrm{com}}^{(1)}, f_{\mathrm{out}}\}$. Training WMMSE-PINet-ImCSI employs the average WSR as the loss function. More specifically, the output of WMMSE-PINet-ImCSI, $\boldsymbol{\theta}$, serves as input to a function that computes the average WSR. In this function, $\boldsymbol{\theta}$ is held constant while multiple groups of channels are sampled. The corresponding WSR for each group of channels is calculated using the WMMSE algorithm to obtain $\mathbf{W}$. Subsequently, the average WSR is computed based on the WSR values of the multiple channel groups.

\section{Simulation Results}
\label{Sec_Result}
\begin{figure}[tbp!]
  \centering
  \includegraphics[width=0.4\textwidth]{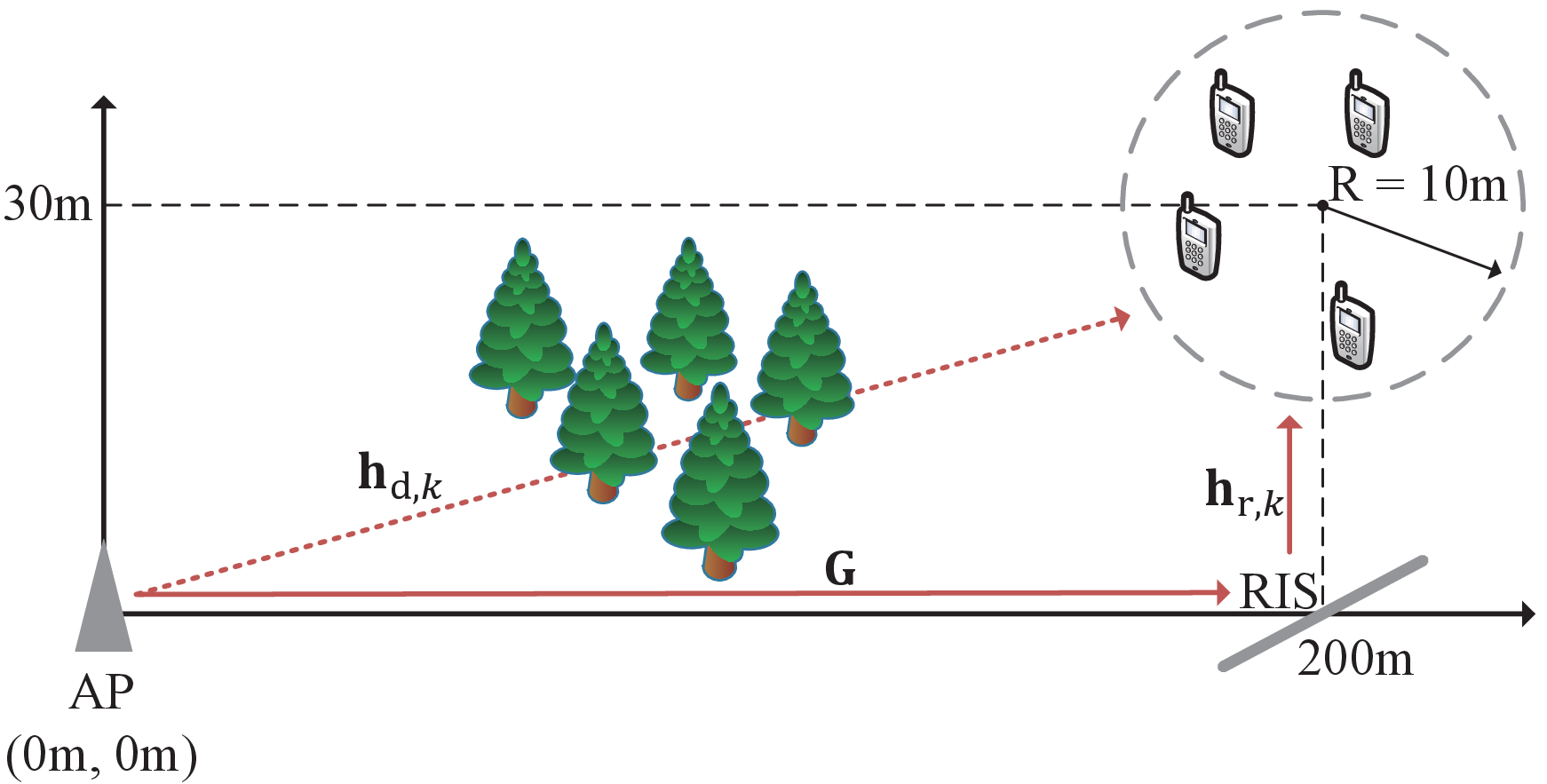}
  \caption{Simulated RIS-assisted downlink MU-MISO system.}
  \label{fig_simulation_env}
\end{figure}
\par
In this section, we present numerical results to illustrate the performance of the proposed algorithms. We begin by presenting specific details concerning the simulation parameters and the process of training the network. Subsequently, we conduct a comprehensive comparison between the performance of the proposed algorithms and that of other state-of-the-art methods across a range of system configurations.

\subsection{Simulation Details}
\par
We simulate a RIS-assisted downlink MU-MISO system as depicted in Fig.~\ref{fig_simulation_env}. Unless otherwise specified, the AP is equipped with $M = 8$ antennas to simultaneously serves $K = 4$ single-antenna users, and the RIS has $N = 100$ reflective elements, which can adjust phase shift continuously. The maximum transmit power is set as $P_\mathrm{T} = 10$ dBm. We assume AP and RIS are located at the origin and ($200\text{ m}$, $0\text{ m}$), respectively. The users are randomly distributed in a circle with a radius of $10\text{ m}$ and a center at ($200\text{ m}$, $30\text{ m}$). The simulation parameters are detailed in Table \ref{tab_parameter}, conforming to the specifications provided in \cite{guoWeightedSumRateMaximization2020} and adhering to the 3GPP propagation environment.

\begin{table}[tbp]
  \caption{Simulation Parameters}
  \label{tab_parameter}
    \centering
    \begin{tabular}{cc}
      \toprule
      Parameters & Values  \\
      \midrule
      Path-loss for $\mathbf{G}$ and $\mathbf{h}_{\mathrm{r},k}$ (dB) & $35.6+22.0$ $\lg \, d$      \\
      Path-loss for $\mathbf{h}_{\mathrm{d},k}$ (dB) & $32.6+36.7$ $\lg \, d$      \\
      Transmission bandwidth & $180$ kHz \\
      Noise power spectral density & $-170$ dBm/Hz \\
      \bottomrule
    \end{tabular}
\end{table}

The channel between AP and users is assumed to be blocked and modeled as Rayleigh fading. The channels from the AP to the RIS and from the RIS to users are postulated to comprise the LOS component and conform to Rician fading, which can be expressed as follows:
\begin{equation}
  \begin{aligned}
  \mathbf{G} &= L_{1}\left(\sqrt{\frac{\kappa}{\kappa+1}} \mathbf{G}_{\rm LOS} +\sqrt{\frac{1}{\kappa+1}} \mathbf{G}_{\rm NLOS} \right), \\
  \mathbf{h}_{\mathrm{r}, k} &= L_{2, k}\left(\sqrt{\frac{\kappa}{\kappa+1}} \mathbf{h}_{\mathrm{LOS, r}, k} +\sqrt{\frac{1}{\kappa+1}} \mathbf{h}_{\mathrm{NLOS, r}, k} \right).
  \end{aligned}
  \label{equ_rician_channel}
\end{equation}

In this equation, $L_1$ and $L_{2,k}$ represent the corresponding path-loss coefficients calculated in accordance with Table \ref{tab_parameter}. The Rician factor is assumed to be $\kappa = 10$. 
$\mathbf{G}_{\rm LOS}$ and $\mathbf{h}_{\mathrm{LOS, r}, k}$ pertain to the LOS components and are described using steering vectors. In contrast, the elements of $\mathbf{G}_{\rm NLOS}$ and $\mathbf{h}_{\mathrm{NLOS, r}, k}$ relate to the NLOS components and are characterized by a standard complex Gaussian distribution.

Throughout the transmission, we assume that the channel statistics for the aforementioned components characterized by the steering vector in (\ref{equ_rician_channel}) remain consistent and are perfectly estimated. Consequently, any channel estimation error is confined to the Rayleigh channels ${\mathbf{h}_{\mathrm{d}, k},  \mathbf{G}_{\rm NLOS}, \mathbf{h}_{\mathrm{NLOS, r}, k}}$. These channels require estimation in every frame.
 We express the corresponding estimated channels as:
	\begin{equation}
		\begin{aligned}
            \widehat{\mathbf{h}}_{\mathrm{d}, k} & = \sqrt{\frac{1}{\varrho+1}} \mathbf{h}_{\mathrm{d}, k} + \sqrt{\frac{\varrho}{\varrho+1}}\mathbf{z}_{\mathrm{d}, k},\\
            \widehat{\mathbf{G}}_{\rm NLOS} & = \sqrt{\frac{1}{\varrho+1}} \mathbf{G}_{\rm NLOS} + \sqrt{\frac{\varrho}{\varrho+1}}\mathbf{Z},\\ 
			\widehat{\mathbf{h}}_{\mathrm{NLOS, r},k} & = \sqrt{\frac{1}{\varrho+1}} \mathbf{h}_{\mathrm{NLOS, r}, k} + \sqrt{\frac{\varrho}{\varrho+1}}\mathbf{z}_{\mathrm{r},k}. 		
		\end{aligned}
	\end{equation}
Here, $\varrho$ represents the normalized mean square error (NMSE) of channel estimation and the smaller $\varrho$ indicates better channel estimation performance. The variables $\mathbf{z}_{\mathrm{d}, k}$, $\mathbf{Z}$, and $\mathbf{z}_{\mathrm{r},k}$ characterize the CSI error, with their elements following a standard complex Gaussian distribution\footnote{We employ a statistical CSI error model, which appears in the objective function as an expectation operation instead of being presented as a constraint, as observed in other works \cite{zhouFrameworkRobustTransmission2020, hongRobustTransmissionDesign2021}.}.
\par
To simply the simulation, received signal noise variances of users are assumed to be equal, that is, $\sigma_1^2 = \sigma_2^2 = \cdots = \sigma_K^2$. The priority for users is determined as the inverse of the path-loss for $\mathbf{h}_{\mathrm{d},k}$. The proposed WMMSE-PINet and WMMSE-PINet$^+$ are trained using the WSR as the loss function, and the proposed WMMSE-PINet-ImCSI is trained using the average WSR as the loss function.
\par
The training dataset and testing dataset contain $1,000$ LOS component samples, and the NLOS component of each sample is generated during the training and testing process. PyTorch, a DL-based architecture, is used for building and training the proposed algorithms, and the networks are trained by using the Adam optimizer. The batch size is set to $50$, and the number of training epochs is $20$. 

\subsection{Convergence Analysis}
\begin{figure}[tb!]
  \centering
  \includegraphics[width=0.45\textwidth, clip]{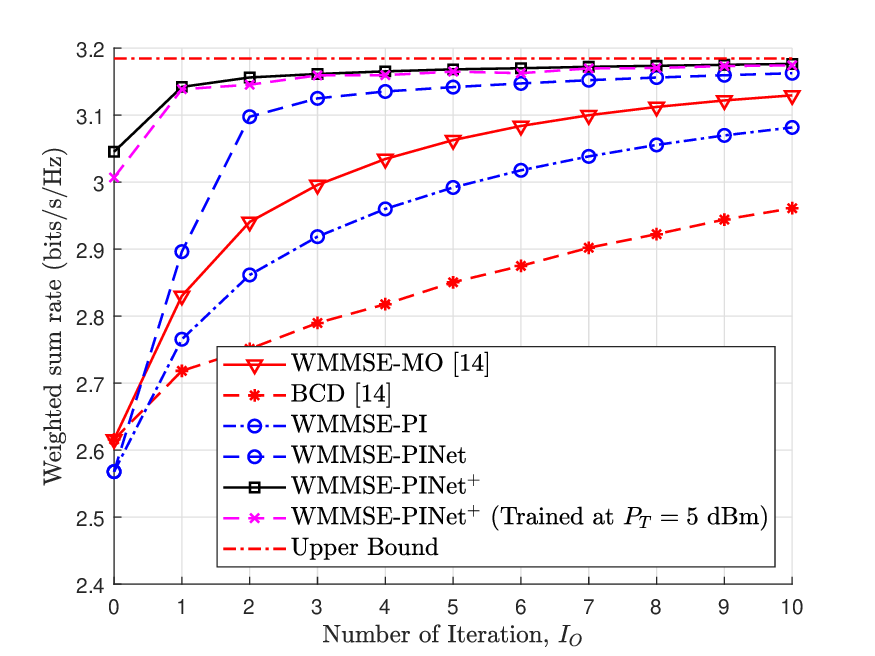}
  \caption{WSR versus the number of iterations, $I_O$.}
  \label{fig_convergence}
\end{figure}

\begin{figure*}[ht!]
  \centering
  \subfigure[]{
  \label{fig_robustness_a}
  \includegraphics[width=0.45\textwidth, clip]{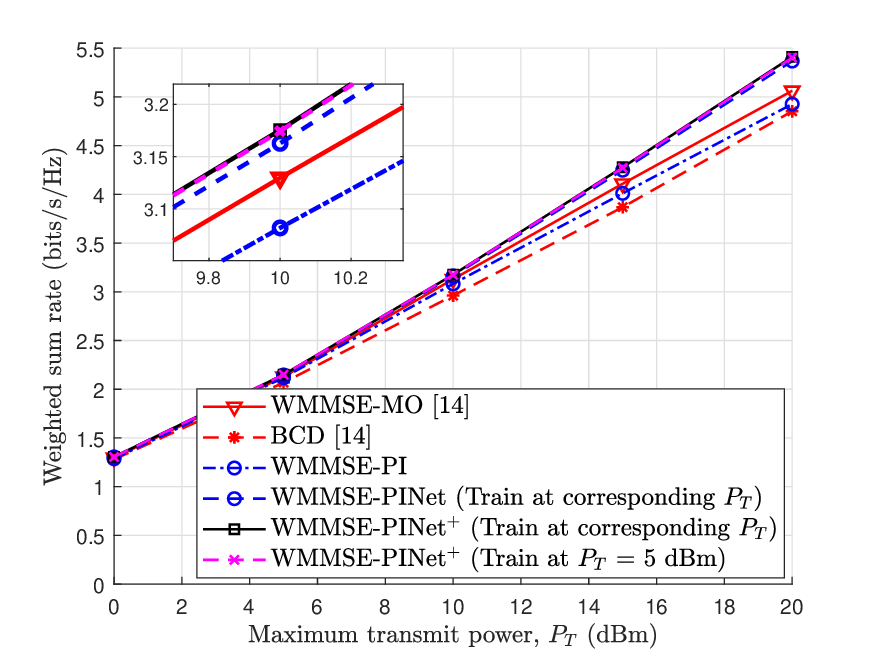}
  }
  \subfigure[]{
  \label{fig_robustness_b}
  \includegraphics[width=0.45\textwidth, clip]{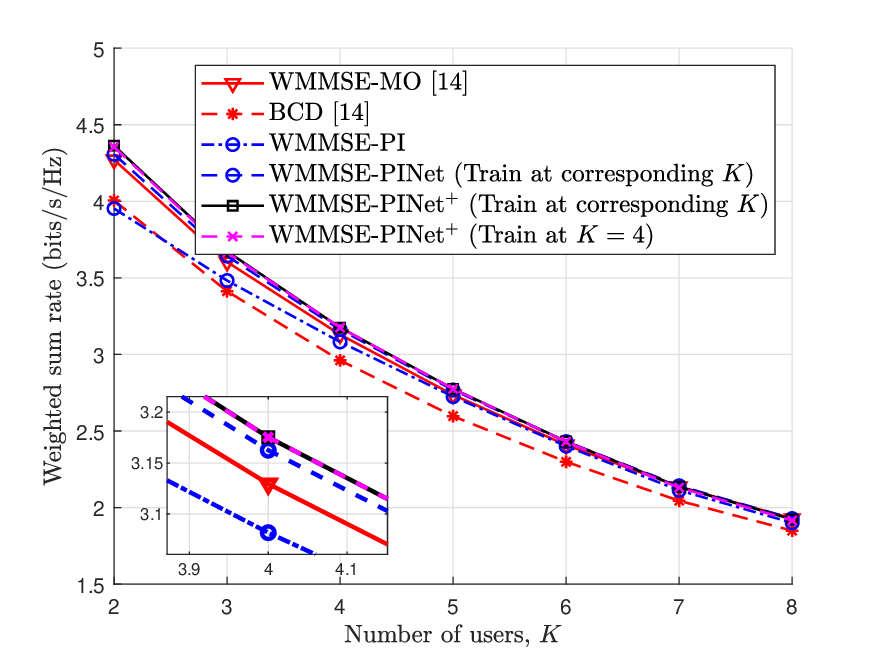}
  }
  \subfigure[]{
  \label{fig_robustness_c}
  \includegraphics[width=0.45\textwidth, clip]{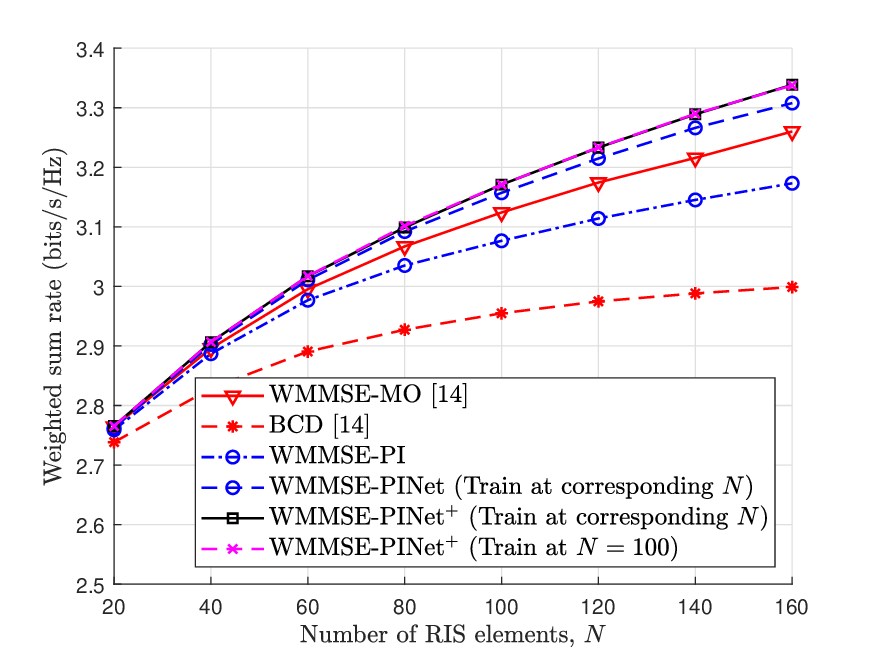}
  }
  \subfigure[]{
  \label{fig_robustness_d}
  \includegraphics[width=0.45\textwidth, clip]{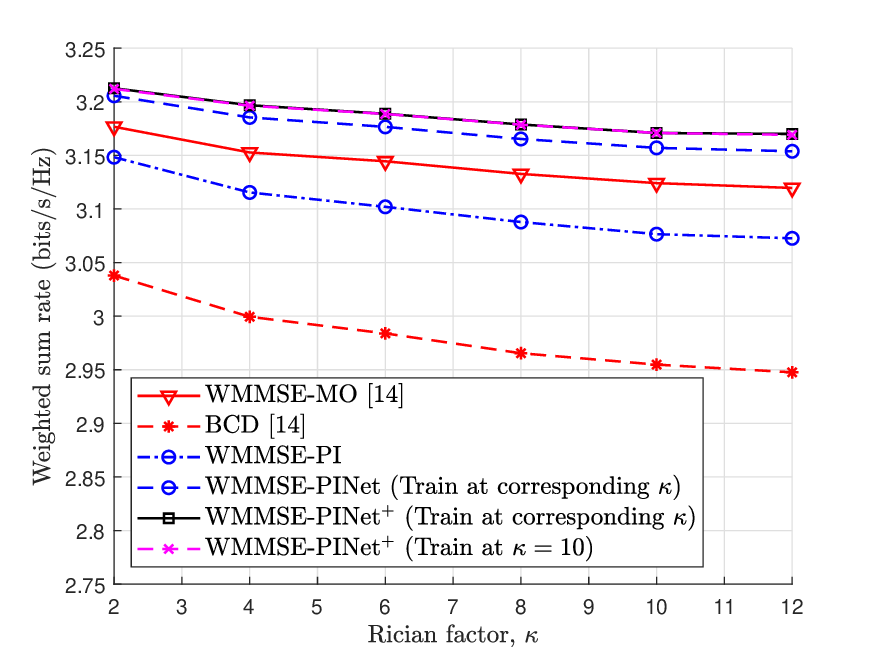}
  }
  \caption{Robustness with different system parameters, (a) maximum transmit power, $P_T$, (b) number of users, $K$, (c) number of RIS elements, $N$, and (d) Rician factor, $\kappa$.}
  \label{fig_robustness}
  \vspace{5pt}
\end{figure*}

In this subsection, we delve into the convergence analysis of the algorithms. Fig.~\ref{fig_convergence} presents the WSR plotted against the number of iterations, denoted as $I_O$. The proposed WMMSE-PI algorithm exhibits superior performance compared to the BCD algorithm. While the WMMSE-MO algorithm displays better performance than the WMMSE-PI algorithm, its complexity, as indicated in Table \ref{tab_complexity}, is notably higher. The WMMSE-PI algorithm strikes a favorable balance between performance and complexity within the scope of the three conventional algorithms.

With the integration of model-driven DL, WMMSE-PINet significantly elevates its performance, surpassing even the more intricate WMMSE-MO algorithm. Leveraging the GNN and the one-iteration PI algorithm for initialization, WMMSE-PINet$^+$ demonstrates further enhancements. The initialization performance is represented as the $0$-th iteration, and remarkably, the initialization of WMMSE-PINet$^+$ attains comparable results to the WMMSE-PI algorithm after $10$ iterations, substantially reducing the required iteration count for convergence.

By virtue of the model-driven DL incorporation, both WMMSE-PINet and WMMSE-PINet$^+$ achieve accelerated convergence rates, converging within merely $5$ iterations. WMMSE-PINet$^+$ yields relatively superior results, largely attributed to its more effective initialization. Notably, running the WMMSE-PI algorithm multiple times (e.g., 100 times) with varying random initializations and employing a sufficiently high number of iterations (e.g., $I_O = 100$) enables us to identify the best-performing solution among them, approximating the global optimal solution. This approach provides an estimate of the upper bound for joint beamforming design, with WMMSE-PINet$+$ with 10 iterations closely approaching this upper bound.

Additionally, the robustness of the proposed DL algorithms is tested by training WMMSE-PINet$^+$ at $P_T = 5$ dBm and evaluating it at $P_T = 10$ dBm. While the initialization performance is slightly affected, subsequent iterations reveal robustness, and the final performances of WMMSE-PINet$^+$ trained under different conditions exhibit similarities. For the ensuing tests, the number of iterations is standardized to $I_O = 10$ for all evaluated algorithms.

\subsection{Robustness Analysis}
\label{sec_robustness}

DL algorithms often exhibit diminished performance when confronted with disparities between training and testing environments. In order to assess the robustness of the proposed algorithms within the context of a RIS-assisted downlink MU-MISO system, we subject them to various scenarios encompassing distinct maximum transmit power $P_T$, numbers of users $K$, numbers of RIS elements $N$, and Rician factor $\kappa$. As depicted in Fig.~\ref{fig_robustness}, the proposed WMMSE-PINet$^+$ consistently outperforms other algorithms across all scenarios, leveraging the inherent advantages of the model-driven DL and the proposed WMMSE-PI algorithm.

Notably, a compelling observation arises when comparing the performance of WMMSE-PINet$^+$ trained at different $P_T$ values and specifically at $P_T = 5$ dBm. The similarity in performance implies a robustness towards variations in the maximum transmit power. This phenomenon can be attributed to the aptitude of WMMSE-PINet$^+$ in grasping the mapping function from channels to beamforming vectors, a task that is arguably accomplished more effectively than conventional data-driven DL approaches. This inherently enhances performance and robustness. Additionally, any slight performance deterioration in the GNN model due to disparities between training and testing environments is mitigated by subsequent iterations.

Similarly, the robustness of the proposed algorithm is substantiated with respect to fluctuations in the number of users, the quantity of RIS elements, and the Rician factor. Consequently, the proposed WMMSE-PINet$^+$ can maintain its performance without necessitating frequent retraining to ensure its efficacy.

\subsection{Phase Shifter Bits}
\begin{figure}[tb]
  \centering
  \includegraphics[width=0.45\textwidth, clip]{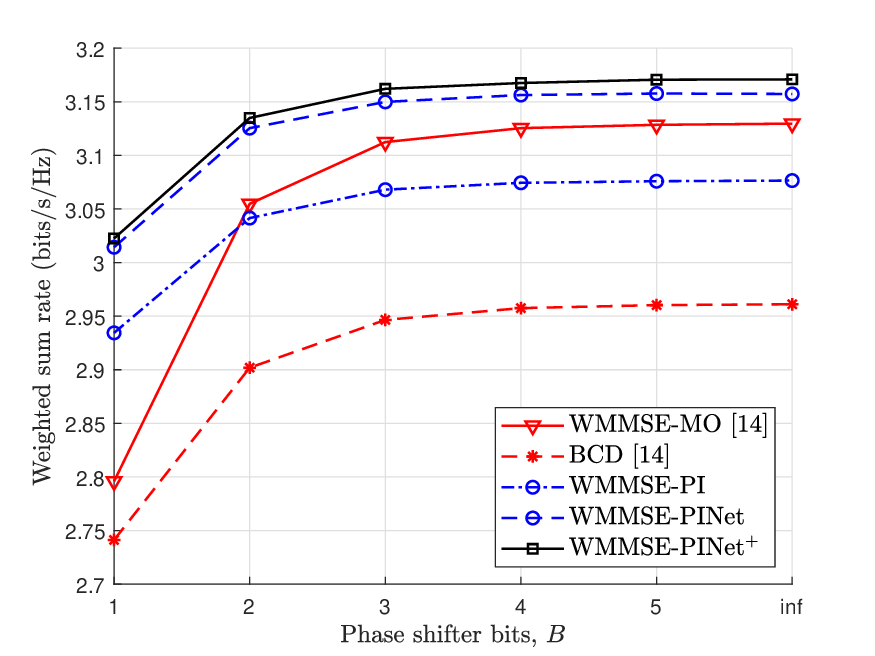}
  \caption{WSR versus the number of phase shifter bits, $B$.}
  \label{fig_Bits}
\end{figure}

This subsection examines the algorithm's performance under a practical scenario where the RIS's phase shifts are restricted to discrete sets. As shown in Fig.~\ref{fig_Bits}, the proposed WMMSE-PINet$^+$ consistently outperforms other algorithms across varying numbers of phase shifter bits. Notably, the WMMSE-PI algorithm also delivers commendable performance, even surpassing the WMMSE-MO algorithm in cases involving just 1-bit phase shifters. The x-label incorporates the value ``inf,'' indicating continuous adjustment of the phase shifter. Interestingly, as the number of phase shifter bits surpasses 3, the performance differences among the various algorithms converge. This observation aligns with the findings of \cite{hanLargeIntelligentSurfaceAssisted2019}, even in the context of our multi-user study. Consequently, it is apparent that an excessive number of phase shifter bits is unnecessary and would merely amplify the energy consumption of the RIS without commensurate benefits.

\subsection{Imperfect CSI}
\begin{figure}[tb]
  \centering
  \includegraphics[width=0.45\textwidth, clip]{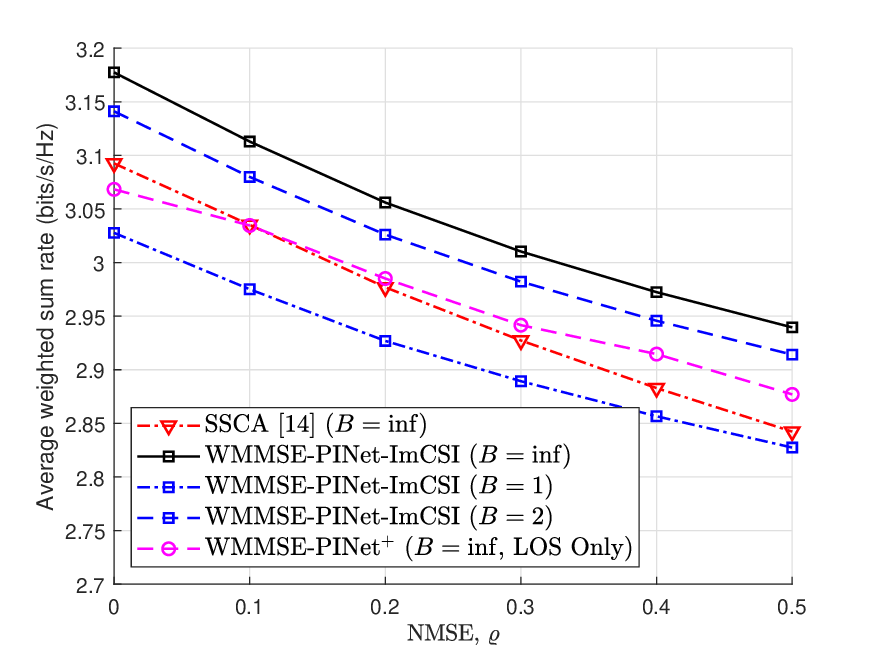}
  \caption{Average WSR versus the NMSE, $\varrho$.}
  \label{fig_ImCSI_varrho}
\end{figure}
\begin{figure}[tb]
  \centering
  \includegraphics[width=0.45\textwidth, clip]{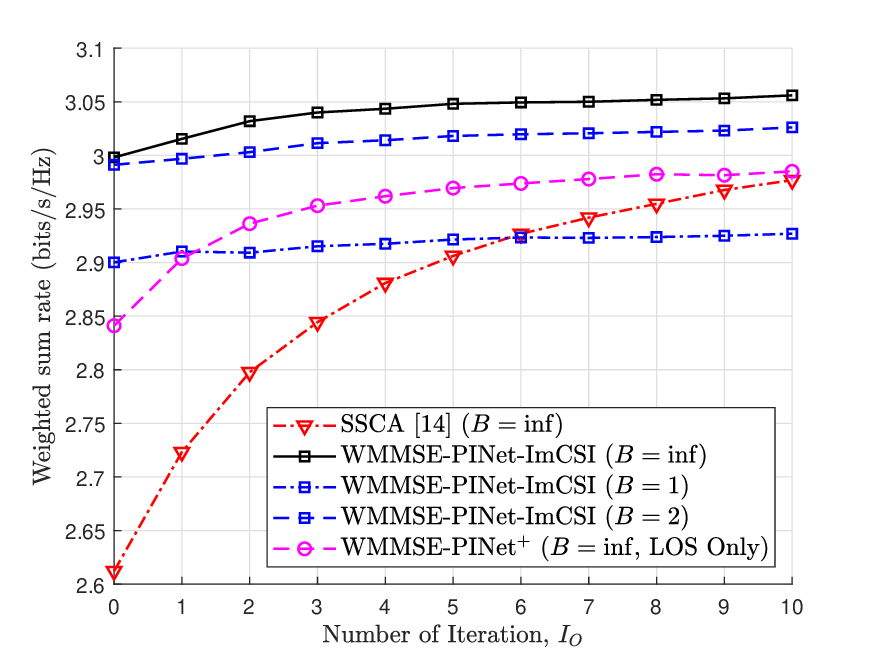}
  \caption{Average WSR versus the number of iterations, $I_O$, with a fixed NMSE of $\varrho = 0.2$.}
  \label{fig_ImCSI_convergence}
\end{figure}

This subsection assesses the performance of the proposed WMMSE-PINet-ImCSI algorithm in the presence of imperfect CSI. As depicted in Fig.~\ref{fig_ImCSI_varrho}, the attained average WSR declines as the NMSE of channel estimation, $\varrho$, increases---a predictable outcome. Notably, WMMSE-PINet-ImCSI with continuous phase shifts outperforms all tested algorithms. Remarkably, even WMMSE-PINet-ImCSI with 2-bit phase shifters exceeds the SSCA algorithm with continuous phase shifts, as proposed in \cite{guoWeightedSumRateMaximization2020}. This highlight the efficacy of our proposed method. As $\varrho$ rises, the performance gap between WMMSE-PINet-ImCSI ($B = 1$) and WMMSE-PINet-ImCSI ($B = \inf$) lessens, indicating that 1-bit phase shifters could improve energy efficiency under suboptimal channel estimation or rapidly fluctuating channels. 

Furthermore, we explore the convergence of the proposed algorithm under imperfect CSI conditions, as shown in Fig.~\ref{fig_ImCSI_convergence}. With the NMSE of channel estimation set at $\varrho = 0.2$, WMMSE-PINet-ImCSI ($B=1$) surpasses the SSCA algorithm ($B = \inf$) in the initial five iterations, illustrating swift convergence due to the initialization with GNN. Consequently, employing 1-bit phase shifters for five iterations can significantly reduce the complexity and hardware requirements of the algorithm.

Given that the NLOS component follows a circularly symmetric complex Gaussian distribution with a mean of zero, its impact on joint beamforming can be mitigated through the expectation operator. Further testing reveals that even when considering only the LOS components of the channel as input (replacing the relevant parts in (\ref{equ_rician_channel}) in the passive beamforming update instead of $\mathbf{G}$ and $\mathbf{h}_{\mathrm{r},k}$), WMMSE-PINet$^+$ still outperforms the SSCA algorithm. Remarkably, WMMSE-PINet-ImCSI ($B = \inf$) and WMMSE-PINet$^+$ ($B = \inf$, LOS only) also exhibit promising performance, with the former benefiting from training considerations under imperfect CSI. However, the channel statistics are influenced by angular parameters and distance, which can be readily obtained using GPS and geometrical relationships. Consequently, channel estimation overhead is significantly reduced for WMMSE-PINet$^+$ (LOS only) since cascaded channels are not necessary in this scenario.

\section{Conclusion}
\label{Sec_Conclusion}

In this study, we have developed efficient joint beamforming techniques for RIS-assisted downlink MU-MISO systems. By splitting the joint beamforming design into active and passive components, we achieved a more manageable optimization problem. The active beamforming vectors were optimized using the WMMSE algorithm, while the passive beamforming vector was iteratively updated using the PI algorithm. To further enhance performance and speed up convergence, we introduced trainable variables and leveraged a GNN to improve the initialization of active beamforming vectors. This improved initialization reduced the required number of iterations for convergence, significantly lowering computational complexity.

We also extended our proposed algorithm to address scenarios with imperfect CSI. Using a two-timescale stochastic optimization approach, we aimed to maximize the average WSR of the system. Our simulation results demonstrated that the proposed algorithm achieved substantial reductions in computational time, performing at less than 3\% of the computational cost of state-of-the-art algorithms. Furthermore, our algorithm exhibited robustness across various system configurations, making it adaptable to changing conditions. Even when not estimating the cascaded channels, our proposed approach outperformed existing algorithms, showcasing its potential for practical deployment.

\appendix
\subsection{Proof of Lemma 1}
\label{Appendix_A}
\par
Based on the definition, we have
\begin{equation}
  \begin{aligned}
    \Re\left\{\left(\boldsymbol{x}^{(p+1)}\right)^H\mathbf{R}\boldsymbol{x}^{(p)}\right\} & =\max _{\left|x_{n}\right|=1} \Re \{\boldsymbol{x}^{H} \mathbf{R}\boldsymbol{x}^{(p)}\}\\
  & \geq \left(\boldsymbol{x}^{(p)}\right)^{H} \mathbf{R x}^{(p)}.
  \end{aligned}
\end{equation}
When $\boldsymbol{x}^{(p+1)}\neq \boldsymbol{x}^{(p)}$ and $\mathbf{R}$ is positive definite, we can obtain
\begin{equation}
  \left(\boldsymbol{x}^{(p+1)} - \boldsymbol{x}^{(p)}\right)^H \mathbf{R}\left(\boldsymbol{x}^{(p+1)} - \boldsymbol{x}^{(p)}\right) > 0.
\end{equation}
Thus, the objective function monotonically increases because
\begin{equation}
  \begin{aligned}
    & \left(\boldsymbol{x}^{(p+1)}\right)^{H} \mathbf{R}\boldsymbol{x}^{(p+1)} \\
    > & 2\Re\Bigl\{\left(\boldsymbol{x}^{(p+1)}\right)^{H} \mathbf{R}\boldsymbol{x}^{(p)}\Bigr\} - \left(\boldsymbol{x}^{(p)}\right)^{H} \mathbf{R}\boldsymbol{x}^{(p)}\\
    > & \left(\boldsymbol{x}^{(p)}\right)^{H} \mathbf{R}\boldsymbol{x}^{(p)}.
  \end{aligned}
\end{equation}
Moreover, the objective function is upper-bounded by
\begin{equation}
  \begin{aligned}
  \boldsymbol{x}^H \mathbf{Rx} & = \sum_{i = 1}^{N+1}\sum_{j = 1}^{N+1} x_i^* R_{i,j} x_j \leq \sum_{i = 1}^{N+1}\sum_{j = 1}^{N+1} |x_i^* R_{i,j} x_j| \\
  & \leq \sum_{i = 1}^{N+1}\sum_{j = 1}^{N+1} |R_{i,j}|.
  \end{aligned}
\end{equation}
Therefore, the PI algorithm is guaranteed to converge to a stationary point $\bar{\boldsymbol{x}}$, which can be denoted as
\begin{equation}
  \mathbf{R}\bar{\boldsymbol{x}} = \operatorname{abs}(\mathbf{R}\bar{\boldsymbol{x}}) \circ e^{\jmath \operatorname{angle}(\mathbf{R}\bar{\boldsymbol{x}})} = \operatorname{abs}(\mathbf{R}\bar{\boldsymbol{x}}) \circ \bar{\boldsymbol{x}}.
\end{equation}
Additionally, we can prove $\operatorname{diag}(\operatorname{abs}(\mathbf{R}\bar{\boldsymbol{x}})) \succeq \mathbf{R}$ as
\begin{equation}
  \begin{aligned}
    &\boldsymbol{x}^H(\operatorname{diag}(\operatorname{abs}(\mathbf{R}\bar{\boldsymbol{x}})) - \mathbf{R})\boldsymbol{x}\\
    = & \boldsymbol{x}^H\operatorname{diag}(\operatorname{abs}(\mathbf{R}\bar{\boldsymbol{x}}))\boldsymbol{x} - \boldsymbol{x}^H \mathbf{R} \boldsymbol{x}\\
    = & \sum_{i = 1}^{N+1} \left|\sum_{j = 1}^{N+1} R_{i,j} x_j\right| - \sum_{i = 1}^{N+1}\sum_{j = 1}^{N+1} x_i^* R_{i,j} x_j\\
    \geq & \left|\sum_{i = 1}^{N+1} \sum_{j = 1}^{N+1} x_i^* R_{i,j} x_j\right| - \sum_{i = 1}^{N+1}\sum_{j = 1}^{N+1} x_i^* R_{i,j} x_j \geq 0.\\
  \end{aligned}
\end{equation}
Thus, $\bar{\boldsymbol{x}}$ is the local optimum of ${\bf P5}$ according to Appendix B of \cite{soltanalianDesigningUnimodularCodes2014}.



\end{document}